\documentclass[12pt]{article}
\usepackage[left=1in,top=1in,right=1in,bottom=1in,nohead]{geometry}

\usepackage{url,subfig,latexsym,amsmath,amssymb, amsfonts,enumerate,
  epsfig, graphics,times, amsthm,mathtools, bbm}

\usepackage{tikz, pgfplots}
\usetikzlibrary{automata,positioning,arrows, intersections, calc}
\pgfplotsset{compat=1.8}

\usepackage{commath}
\usepackage{wrapfig}
\usepackage{xcolor}
\usepackage{nameref}
\usepackage{listings}
\usepackage[colorlinks]{hyperref}
\usepackage{makecell}
\usepackage{makeidx}
\usepackage{tabularx}
\usepackage{array}
    \newcolumntype{P}[1]{>{\centering\arraybackslash}p{#1}}
    \newcolumntype{M}[1]{>{\centering\arraybackslash}m{#1}}

\usetikzlibrary {graphs}
\usepackage{chemfig}
\usepackage{epstopdf}
\usepackage{svg}
\usepackage{stmaryrd}
\newtheorem{thm}{Theorem}[section]

\newtheorem{lem}[thm]{Lemma}
\newtheorem{prop}[thm]{Proposition}

\newtheorem*{thm*}{Theorem}

\theoremstyle{definition}
\newtheorem{defn}{Definition}[section]
\newtheorem{rem}{Remark}[section]

\newcommand{\dlim}{\displaystyle \lim\limits}

\def\T1{T^1_{\{x_n\}}}
\def\Sp{\mathcal S}
\def\C{\mathcal C}
\def\Re{\mathcal R}

\def\Td1{T^{D,1}_{\{x_n\}}}
\def\Ts1{T^{S,1}_{\{x_n\}}}

\def\K{\mathcal{K}}

\def\Tid1{T^{D,1}_{\{x_n+\zeta_i\}}}
\def\Tis1{T^{S,1}_{\{x_n+\zeta_i\}}}
\def\Tjs1{T^{S,1}_{\{x_n+\zeta_j\}}}
\def\Tjd1{T^{D,1}_{\{x_n+\zeta_j\}}}

\newcount\colveccount
\newcommand*\colvec[1]{
        \global\colveccount#1
        \begin{pmatrix}
        \colvecnext
}
\def\colvecnext#1{
        #1
        \global\advance\colveccount-1
        \ifnum\colveccount>0
                \\
                \expandafter\colvecnext
        \else
                \end{pmatrix}
        \fi
}

\makeatletter
\newcommand{\labitem}[2]{%
\def\@itemlabel{\textbf{#1}}
\item
\def\@currentlabel{#1}\label{#2}}
\makeatother

\title{
 Noise-robust chemical reaction networks training artificial neural networks}

\author{
Sung-Hwa Kang\thanks{...},
\and 
Jinsu Kim\thanks{Department of Mathematics, Pohang University of Science Technology, Pohang 37673, Republic of Korea. jinsukim@postech.ac.kr, grant support from the National Research Foundation of Korea (NRF) grant funded by the Korea government (MSIT)(No. 2022R1C1C1008491). }
  }

\begin{document}

\maketitle

\begin{abstract}
Artificial neural networks (NNs) can be implemented using chemical reaction networks (CRNs), where the concentrations of species act as inputs and outputs. In such biochemical computing, noise-robust computing is crucial due to the intrinsic and extrinsic noise present in chemical reactions. Previously suggested CRNs for feed-forward networks often utilized the rectified linear unit (ReLU) or discrete activation functions. However, one concern in this case is the discontinuities of the derivatives of those non-smooth functions, which can cause significant noise disruption during backpropagation.
In this study, we propose a CRN that performs both feed-forward and training processes using smooth activation functions to avoid  discontinuities in the backpropagation. All reactions occur in a single pot, and the reactions for training are bimolecular. Our case studies on XOR, Iris, MNIST datasets, and a non-linear regression model demonstrate that computation via the CRN (i) maintains accuracy despite noise in the reaction rates and the concentration of species and (ii) is insensitive to the choice of the running time and the magnitude of the noise in comparison to NNs with a non-smooth activation function. 
This work presents a noise-robust CRN for full NN computation, including backpropagation, paving the way for more stable and efficient biochemical computing systems.
\end{abstract}

\tableofcontents


\section*{Significance Statment}
Chemical reaction networks describe how the counts or concentrations of chemical species change. By manipulating biochemical reactions, we can use them as a suitable language for coding computational algorithms like artificial neural networks (NNs). We have developed a dynamical system based on mass-action chemical reaction networks (CRNs) to facilitate the training of NNs equipped with smooth activation functions. Our proposed CRN implements both the feed-forward and backward processes of an NN. We have demonstrated that this CRN shows improved noise robustness compared to previously introduced CRNs that implement NNs with non-smooth activation functions.

\section{Introduction}
This paper explores the use of chemical reaction networks (CRNs) for implementing neural networks (NNs). In a broad sense, implementing NNs through biological systems represents a form of biological computing. Biological computers are designed using biologically derived objects such as DNA and proteins to perform digital computations, which are typically carried out by silicon chip-based classical computers. Due to their parallel processing abilities and energy efficiency, biological computation has attracted significant interest \cite{ezziane2005dna}. Advances in synthesizing specific chemical reactions using DNA strands and enzymes have enabled the development of biological computers capable of handling more intricate tasks \cite{soloveichik2010dna}.

Similar to programming languages used in classical computers, chemical reactions can act as programming languages in biological computing. Studies conducted both in simulation and in laboratory settings have demonstrated that chemical reaction networks (CRNs) can perform computations similar to those of artificial neural networks \cite{hjelmfelt1991chemical, hjelmfelt1992chemical, anderson2021reaction, xiong2022molecular, kim2004neural, vasic2022programming, lakin2014design, arredondo2022supervised, qian2011neural, blount2017feedforward}. In these studies, CRNs were designed to reach a unique stable steady state when the initial concentrations of certain chemical species are set as input, effectively producing the same output as neural networks. The dynamics of these chemical species are governed by a system of ordinary differential equations (ODEs) of the form $\dot x(t) = f(x(t); \hat x)$, where $x(t)$ represents the concentrations of the species over time and $\hat x$ is the input. This process can be viewed as a computing process, where setting the initial concentrations of certain species in the CRN to the input $\hat x$ of the neural network leads to the concentrations of the other species converging to a specific steady state $y(\hat x)$ as the desired output of the neural network.



While most of the previously designed chemical reaction networks for neural network computation aimed at performing the feed-forward network, few works focused on the training neural network via chemical reactions \cite{banda2013online, arredondo2022supervised, lakin2023design}. In \cite{banda2013online}, it was shown that a learning algorithm for binary logic gates can be realized with chemical reactions that automatically update the weight parameters of neural networks based on the accuracy of the outputs.
In \cite{arredondo2022supervised}, a chemical reaction network for supervised learning was proposed, where the hyperbolic tangent activation function was realized via chemical reactions and the networks are trained by a perturbation algorithm. In \cite{lakin2023design}, the authors considered a NN with the leaky ReLU, a non-smooth activation function. The NN is trained via chemical reactions that compute the piecewise derivatives of the leaky ReLU.

In previous research, activation functions like ReLU, leaky ReLU, and binary activation have been popular choices for chemical reaction computing due to their simple engineering. However, using these activation functions in chemical reactions may lead to potential issues with noise robustness, particularly during backpropagation, as these functions have discontinuities in their derivatives. Unlike standard neural networks, chemical reactions can experience random fluctuations in reaction rate parameters due to changes in temperature, pressure, and other external disturbances. These fluctuations can introduce small perturbations to each node of the neural network. If the activation function is non-smooth, these perturbations may amplify errors during backpropagation due to the discontinuities in their derivatives.
For example, the derivative $\sigma'(z)$ of ReLU ($\sigma(z)=z$ for $z\ge 0$ and $\sigma(z)=0$ otherwise) should be $1$ when $z=0.01$. However, even a small noise can change this value to be $-0.01$, resulting in a derivative of $0$, and potentially leading to significant errors. While experimental results have validated the accuracy of chemical reactions in performing the feed-forward process of neural networks, the use of non-smooth activation functions raises concerns about noise robustness, especially if the entire neural network computation, including backpropagation, is carried out through chemical reactions.

In a recent study, an interesting concept was proposed that involves using binary expression to represent the presence of species in a steady state, which determines the parameters of a neural network \cite{vasic2022programming}. This method relies solely on the stoichiometric structures of the chemical reaction network, rather than the rate parameters. The study demonstrated that the weights of neural nodes can be set as binary or rational numbers, and these weights are determined by the stoichiometric structures, independent of the reaction rates. Therefore, if the training process is also carried out using chemical reactions, the stoichiometric structures of the feed-forward network would need to be updated by chemical species. However, implementing such procedures in chemical engineering is likely to be challenging, and it remains uncertain how this method would perform in the presence of noise.

The main goals of this paper is to provide a design of a chemical reaction network whose associated dynamical system trains a given artificial neural network equipped with smooth activation functions. This chemical reaction network is capable of performing full neural network computations, including NN parameter training (Figure \ref{fig:main schematic fig}). The reactions take place in a single container, allowing for both feed-forward processing and training to occur simultaneously. We have utilized the CRN proposed in \cite{anderson2021reaction} for the feed-forward part of the NN. Additionally, we developed a new CRN specifically for the backpropagation process. This new network plays a crucial role in computing the derivatives of smooth activation functions. The essential role of this CRN is to compute the derivatives of a smooth activation function. The building blocks for our CRN are the chemical reactions proposed in \cite{buisman2009computing} that compute basic arithmetic operations such as addition, subtraction, multiplication, and division.

In our research, we conducted case studies using the propose CRNs for the XOR, Iris and MNIST datasets, as well as non-linear regression models. Our findings demonstrated that CRNs implementing a neural network with smooth activation functions are resilient to noise in the rate constants and the species corresponding to nodes. This resilience is also supported by theoretical validation. In comparison, we have observed that CRNs implementing non-smooth activation functions can result in significant errors with small perturbations in the rate parameters due to the discontinuity of the derivatives of the non-smooth activation functions. Additionally, our research showed that our CRNs are more robust to variations in noise levels and the duration of the associated chemical reaction systems.



\begin{figure}[!h]
    \centerline{\includegraphics[width=15cm]{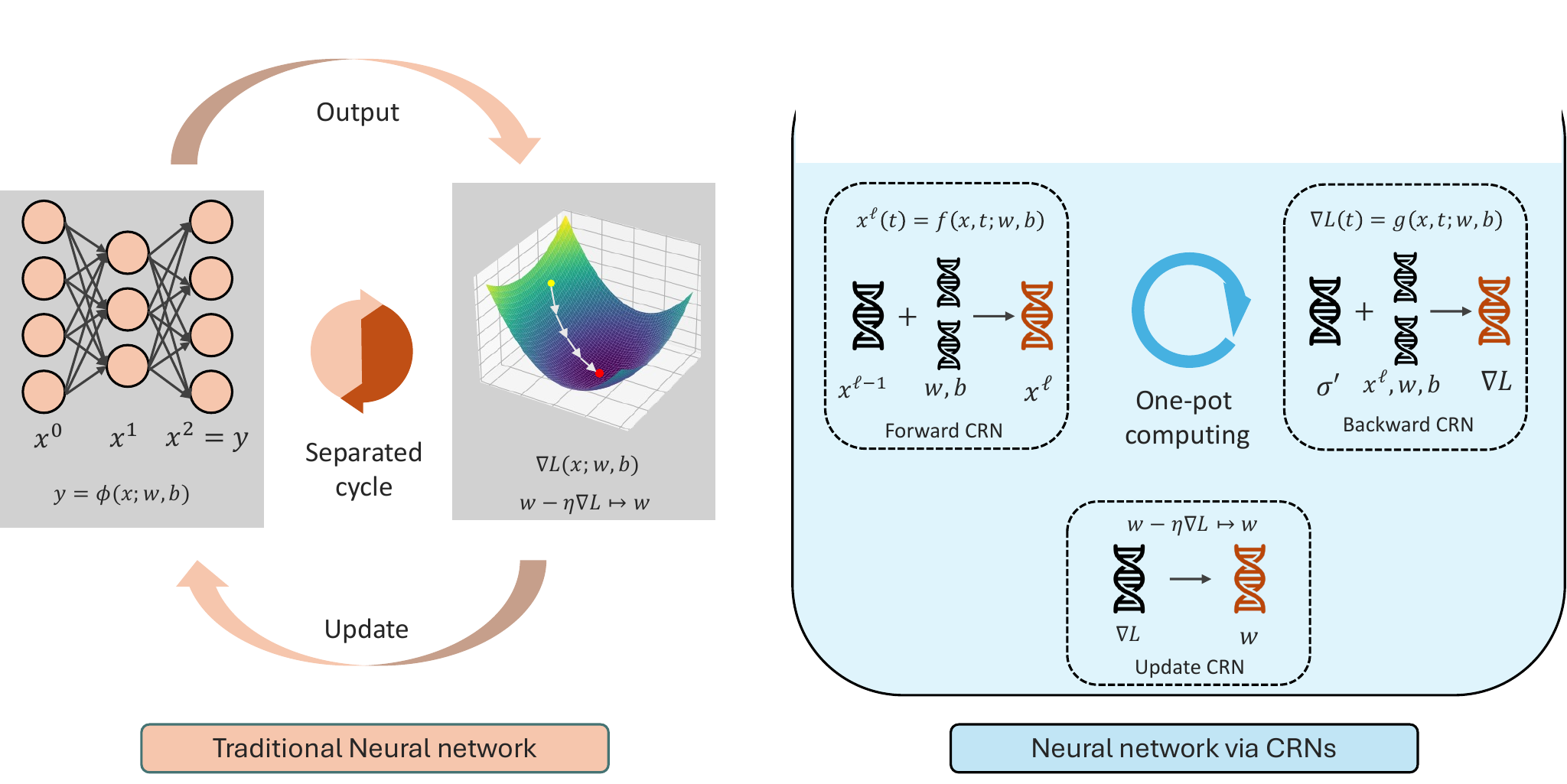}}
    \caption{Abstract of our work.}\label{fig:main schematic fig}
\end{figure}

\section{The key concepts}
In this paper, we used two different types of networks: chemical reaction networks and neural networks. To avoid confusion, we define them more precisely in this section. 
\subsection{Chemical reaction networks}

A chemical reaction network is a graph where the nodes represent complexes composed of combinations of chemical species, and the edges are reactions between complexes. Figure \ref{ex:crn1} shows an example of  chemical reaction networks.
\begin{figure}[!h]
    \centering
   \tikzset{state/.style={inner sep=2pt}}
   \begin{tikzpicture}[baseline={(current bounding box.center)}, scale=0.8]
   \node[state] (1) at (-.5,2)  {$H$};
   \node[state] (2) at (2,2)  {$X_1+H$};
   \node[state] (3) at (5,2)  {$X_2+H$};
   \node[state] (4) at (8,2)  {$2X_1+W$};
   \node[state] (4-1) at (8.5,2.3) {};
   \node[state] (4-2) at (8.5,1.7) {};
   \node[state] (5) at (11,2)  {$W$};
   \node[state] (6) at (13,2)  {$X_1$};
   \node[state] (7) at (16,2)  {$X_1+X_2$};
   \node[state] (8) at (14.5,0.5)  {$\emptyset$};
   \path[->]
    (2) edge node[above]{$k_2$} (3)
    (1) edge node[above] {$k_1$} (2)
    (1) edge[bend right=25] node[below] {$k_3$} (3)
    (4-1) edge[bend left] node[above] {$k_4$} (5)
    (5) edge[bend left] node[below] {$k_5$} (4-2)
    (6) edge node[above] {$k_6$} (7)
    (7) edge node[below] {$k_7$} (8)
    (8) edge node[below] {$k_8$} (6);
  \end{tikzpicture}
    \caption{An example of chemical reaction networks}
    \label{ex:crn1}
\end{figure}
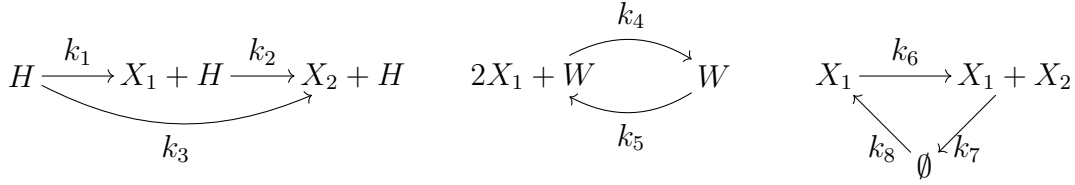

A precise mathematical formulation is as follows.
\begin{defn}\label{def:21}
\emph{A chemical reaction network} is given by finite sets $(\Sp,\C,\Re,\mathcal K)$ such that
\begin{enumerate}
\item  $\Sp=\{S_1,S_2,\cdots,S_d\}$ is a set of $d$ symbols, called the \emph{species} of the network,
\item $\C$ is a set of $\nu=\sum_{i=1}^d \nu_iS_i$, a linear combinations of $S_i$'s with $(y)_i\in \mathbb Z_{\ge 0}$ for each $i$, called \emph{complexes},
\item $\Re$ is a subset of $\C\times\C$, whose elements $(\nu,\nu')$ are typically denoted by $\nu\to \nu'$,
\item $\mathcal K$ is a set of reaction rates $\{k_1,\dots,k_{|\Re|}\}$, where each $k_i$ is assigned to each reaction. 
\end{enumerate}
\end{defn}

$\nu_i$'s in each complex $\nu=\sum_{i=1}^d \nu_iS_i$ are stoichiometric coefficients. The `zero' complex $\emptyset$ represents the external environment so that the birth and death of species are described as $\emptyset\to X$ and $X\to \emptyset$, respectively.  

 In the case of the chemical systems being spatially well-mixed, the associated dynamics can be defined as a system of ODEs on $\mathbb R^d_{\ge 0}$, which models  the time evolution of the concentrations of species in the chemical reaction network. 
The concentration of each species is modeled with $x(t)=(x_1(t),\dots,x_d(t))$ that is a solution to 
\begin{align}\label{eq:system of odes}
    \frac{d}{dt}x(t)=\sum_{\nu\to \nu'\in \Re} \lambda_{\nu\to \nu'}(x(t))\eta_{\nu\to \nu'}.
\end{align}
The stoichiometric vector $\eta_{\nu\to \nu'}$ gives the net gain of the reaction. For example, for the reaction network in Figure \ref{ex:crn1}, the stoichiometric vector of $2X_1+W\to W$ is $(-2,0,0,0)^\top$ as the system loses two $X_1$ species via this reaction while other species are unchanged. 
While kinetics-independent chemical reaction computing was suggested \cite{vasic2022programming}, our training processes via chemical reactions use mass-action, \begin{align}\label{mass}
\lambda_{\nu\rightarrow \nu'}(x)=\kappa_{\nu\rightarrow \nu'} \prod_{i=1}x_i ^{\nu_i},
\end{align}
for the reaction rate $\kappa_{\nu\to \nu'}$ assigned to $\nu\to \nu'$.

To calculate derivatives of the activation function, our chemical computing uses catalytic reactions. Catalytic reactions are known to be required to compute non-linear functions \cite{okumura2022nonlinear}. For the sake of clarity, we define two classes of species as suggested in \cite{anderson2021reaction}.
\begin{defn}
For a reaction network $(\Sp,\C,\Re,\mathcal K)$, we define two mutually disjoint sets $\Sp_c$ and $\Sp_d$ such that 
\begin{enumerate}
        \item $\Sp=\Sp_c \bigcup \Sp_d$,
    \item for each $S_i\in \Sp_c$, the $i$ the component of each stoichiometric vector is zero, i.e. we have $(\eta_{\nu \to \nu'})_i=0$ for any $\nu\to \nu'\in \Re$.
\end{enumerate}
We call species in $\Sp_c$ \emph{catalytic species} and call species in $\Sp_d$ \emph{dynamics species}.
\end{defn}
For example, in the reaction network \eqref{ex:crn1}, $\Sp_c=\{W,H\}$ and $\Sp_d=\{X_1, X_2\}$. 
Note that the concentrations of catalytic species remain unchanged unlike dynamic species. 

\begin{rem}
The role of species in the CRN for the full neural network is interchangeable.
The CRN consists of three sub-networks: $(\Sp_f,\C_f,\Re_f,\mathcal K_f)$, $(\Sp_b,\C_b,\Re_b,\mathcal K_b)$ and $(\Sp_u,\C_u,\Re_u,\mathcal K_u)$ implement the feed-forward neural network, the backpropagation, and updating the neural network parameters, respectively.
In $(\Sp_f,\C_f,\Re_f,\mathcal K_f)$, we use the concentrations of the catalytic species for either the inputs or the parameters of the feed-forward neural network computation \cite{anderson2021reaction}. However, in $(\Sp_u,\C_u,\Re_u,\mathcal K_u)$, the species that were catalytic in the feed-forward sub-network become dynamic species, allowing the parameters in the feed-forward part to be updated for training the neural network.
\end{rem}

\subsection{Artificial neural networks}\label{sec:nn}

 In mathematical terms, a NN is a combination of a weighted sum, where each weighted sum is connected to an activation function. The feed-forward part of the NN computes the composition of activation functions with the weighted sum. Common activation functions include sigmoidal functions, hyperbolic tangent, and ReLU functions. This paper primarily uses the smoothed ReLU function as an activation function, which is defined as
\begin{align}\label{eq:smooth relu}
    \sigma(z) = \frac{z+\sqrt{z^2+4h}}{2},
\end{align}
where $h$ determines the smoothness of $\sigma(x)$ around $x=0$ (Figure \ref{fig:activation and nn} A). Note that if $h=0$, then $\sigma$ is the standard ReLU function. 

 The backpropagation process involves updating the weights of the feedforward part by using the gradient of the loss function with respect to the weights, based on the chain rule \cite{rumelhart1986learning}. One of the main results of this paper is the development of a CRN that implements the backpropagation algorithm for a NN with the smooth ReLU activation function. We will compare the performance of our CRN with the CRN that implements the backpropagation algorithm for a NN with the leaky ReLU activation function (Figure \ref{fig:activation and nn} B), defined as 
\begin{align}\label{eq:leaky Relu}
    \sigma_{Leaky}(z)=\begin{cases}
        \alpha z \quad &\text{if $z\ge 0$}\\
        \beta z &\text{if $z<0$},
    \end{cases}
\end{align}
where $0<\beta\ll 1$. This activation function was utilized in \cite{lakin2023design} for training a NN via CRNs.

We will start by providing precise notations for NNs used to explicitly describe the CRNs carrying out NNs in the later sections.
Let $G=(V,D)$ denote a NN, where $V$ and $D$ are the set of nodes and edges in the neural network, respectively. The concepts for feed-forward neural networks are illustrated below.  Please refer to Figure \ref{fig:activation and nn} C for a graphical representation.

\begin{enumerate}
\item The first and the last layers referred to as the input and output layers, respectively.Any layers in between are called hidden layers.

\item The value assigned to the $i$ th node (neuron) on the $\ell$ th layer is denoted by $x^\ell_i$.      
    \item The number of the nodes on the $\ell$ th layer is denoted by $M^\ell$, which represents the width of the $\ell$ th layer.
    \item  We sometimes use a vector $x^\ell=(x^\ell_1,\dots,x^\ell_{M^\ell}) \in \mathbb R^{M^\ell}$.
    \item The function $\sigma:\mathbb R\to \mathbb R$ denotes the activation function. Abusing a notation, we define a multi-variable function $\sigma:\mathbb R^m\to \mathbb R^m$ for $m$ such that the $i$ th component of $\sigma(z)\in \mathbb R^m$ is $\sigma(z_i) \in \mathbb R$ for a $z\in \mathbb R^m$. 
    
    \item Let $w^{\ell} \in \mathbb R^{M^{\ell+1}\times M^{\ell}}$ denote the weight matrix and let $b^\ell \in \mathbb R^{M^{\ell+1}}$ denote the bias vector. Then we define $z^\ell = w^\ell x^\ell +b^\ell \in  \mathbb R^{M^{\ell+1}}$ and $x^{\ell+1}=\sigma(z^\ell) \in \mathbb R^{M^{\ell+1}}$ for $\ell\in\{0,1,\dots,L-1\}$, where $\sigma(z^\ell)_i=\sigma(z^\ell_i)$ for $i\in\{1,2,\dots,M^{\ell+1}\}$. 
    \item The elements of $w^\ell$ and $b^\ell$ are denoted by $w^\ell_{ij}$ and $b^\ell_i$. 
    \item The values $x^1_i$ and $x^L_i$ assigned to the $i$ th node in the input layer $V^1$ and the output layer $V^L$ are specially denoted by $\hat x_i$ for $i\in \{1,\dots,M^1\}$ and $y_i$ for $i\in \{1,\dots,M^L\}$, respectively.  Hence $x^1=\hat x\in \mathbb R^{M^1}$ and $x^L=y\in \mathbb R^{M^L}$.
    
\end{enumerate}
If we let $x^{\ell+1}=\psi^{\ell+1}(x^\ell):=\sigma(z^\ell)$ for $\ell=0,1,\dots, L-1$, the output of the neural network is simply written as 
\begin{align*}
    y=\psi^L \circ \psi^{L-1} \circ \cdots \circ \psi^{3} \circ \psi^2(a).
\end{align*}
For a given feed-forward neural network $(V,D)$ with the input $\hat x \in \mathbb R^{M^1}$ and the output $y\in \mathbb R^{M^L}$, for instance, a CRN $(\Sp,\C,\Re,\mathcal K)$ can be constructed so that the associated deterministic model $x(t)=(x_d(t), x_e(t))$ such that
\begin{align*}
    \frac{d}{dt} \begin{pmatrix}
        x_d(t) \\
        x_c(t)
    \end{pmatrix} = f(x_d(t), x_c(t))
\end{align*}
satisfies $\dlim_{t\to \infty}x_d(t)=y$ with $x_c(0)=\hat x$, where $x_d$ and $x_c$ represent the concentrations of the dynamical and the catalytic species, respectively.

\begin{figure}[!h]
    \centering
\includegraphics[width=1\linewidth]{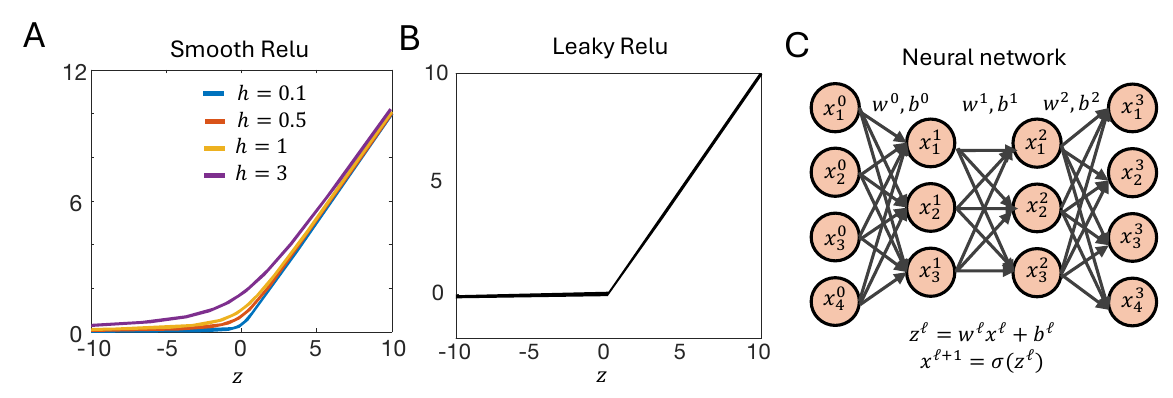}
    \caption{A. The graph of the smooth ReLU with different smoothing parameters $h$. B. The graph of the leaky ReLU. C. An example of neural networks.}
    \label{fig:activation and nn}
\end{figure}

\subsubsection{Gradient of the loss function and derivatives of smooth activation functions.}\label{sec: gradient descent}
Training NNs is a procedure to find a weight matrix and a bias vector that minimize the error of the feed-forward neural network. The error can be measured using a loss function. 
Throughout this manuscript, we define the loss function defined as $\mathcal L(y)=|y-\hat y|^2$, where $y \in \mathbb R^{M^L}$ is the output of the neural network and $\hat y \in \mathbb R^{M^L}$ is the output of the training data sets.

When minimization is performed with the (stochastic) gradient descent algorithms, the main task is to compute the partial derivatives of the loss function $\mathcal L(y)=|y-\hat y|^2$ with respect to the parameters of the neural network $w^\ell_{ij}$'s and $b^\ell_i$'s. Backpropagation is typically initiated from the output layer $\ell=L$. For suitable indices $i$ and $j$ we have that
\begin{align}
&\frac{\partial \mathcal L}{\partial x^{L}_{i}}=\frac{\partial \mathcal L}{\partial y_{i}}=2(y_i-\hat y_i), \quad \frac{\partial \mathcal L}{\partial z^{\ell}_{i}} = \sum_j \frac{\partial \mathcal L}{\partial x^{\ell+1}_{j}}  \frac{\partial x^{\ell+1}_{j}}{\partial z^{\ell}_{i}} \quad \text{and}\notag \\
& \frac{\partial \mathcal L}{\partial x^{\ell}_{i}}
=\sum_j \frac{\partial \mathcal L}{\partial z^{\ell}_{j}}  \frac{\partial z^{\ell}_j}{\partial x^{\ell}_{i}} 
=\sum_j \frac{\partial \mathcal L}{\partial z^{\ell}_{j}}  w^\ell_{ji}, \label{eq:L derivative x}
\end{align}
for $\ell \in \{1,...,L-1\}$. Then
 gradients of $\mathcal L$ can be recursively written as
\begin{align}
&\nabla_{x^\ell} \mathcal L = \left(W^\ell\right)^T\left(\nabla_{x^{\ell +1}}\mathcal L \odot \nabla_{z^\ell} x^{\ell+1}\right) = \left(W^\ell\right)^T\left(\nabla_{z^{\ell}} \mathcal L\right),\label{eq:derivative of L wrt x}\\
&\nabla_{w^\ell} \mathcal L = \left(\nabla_{x^{\ell+1}} \mathcal L \odot \nabla_{z^\ell} x^{\ell+1}\right)\left(x^{\ell}\right)^{T} = \left(\nabla_{z^{\ell}} \mathcal L\right)\left(x^{\ell}\right)^{T} ,\label{eq:derivative of L wrt w} \\
&\nabla_{b^\ell}\mathcal L=\left( \nabla_{x^{\ell+1}}\mathcal L  \odot \nabla_{z^\ell} x^{\ell+1}\right) = \nabla_{z^{\ell}} \mathcal L \label{eq:derivative of L wrt b},
\end{align}
where $\nabla_{z^\ell} x^{\ell+1}$ is a vector whose $j$ th component is $\dfrac{\partial \sigma(z^\ell)_j}{\partial z^\ell_j}$, and $\odot$ is the Hadamard product \cite{horn1990hadamard}, which stands for term-wise multiplication.
 Regarding the recursive relation \eqref{eq:derivative of L wrt x}--\eqref{eq:derivative of L wrt b}, the key term of the backpropagation is $\nabla_{z^\ell} x^{\ell+1}$ for each $\ell$.

The derivative of the smoothed ReLU function $\sigma(z)=\frac{z+\sqrt{z^2+4h}}{2}$ is 
\begin{equation}\label{eq:derivative of relu}
\sigma'(z) = \frac{z+\sqrt{z^2+4h}}{2\sqrt{z^2+4h}} = \sigma(z)^2\frac{1}{\sqrt{z^2+4h}}\frac{2}{z+\sqrt{z^2+4h}}=\frac{\sigma(z)^2}{\sigma(z)^2+h}.
\end{equation}
Then if $x^{\ell+1}_j=\sigma(z^\ell_j)$ for each $\ell$ and $j$, then
\begin{align*}
    \nabla_{z^\ell}x^{\ell+1}=\left (\frac{(x^{\ell+1}_1)^2}{(x^{\ell+1}_1)^2+h},\dots,\frac{(x^{{\ell+1}}_{M^\ell})^2}{(x^{\ell+1}_{M^\ell})^2+h} \right )
\end{align*}
Remarkably, these derivatives can be simply computed by addition and division of the nodes $x^\ell_i$'s that are computed in the feed-forward part of a NN.

\section{CRNs for training}
In this section, we will discuss the construction of Chemical Reaction Networks (CRNs) used for computing the gradients of the loss function $L(y)=|y-\hat y|^2$ and for updating the parameters accordingly.
The CRNs responsible for computing gradients will be referred to as the "backward CRN," and those responsible for updating parameters as the "update CRN." We will denote the backward CRN and the update CRN by $(\Sp^b, \C^b, \Re^b,\K^b)$ and $(\Sp^u, \C^u, \Re^u,\K^u)$, respectively.
For the complete neural network computation, we will utilize the CRN provided in \cite{anderson2021reaction}, which implements a feed-forward NN with the smooth ReLU function. This CRN will be referred to as the "forward CRN" and will be denoted by $(\Sp^f, \C^f, \Re^f,\K^f)$.
To maintain consistency, we will use capital letters to represent species and small letters to denote their concentrations. For example, the species corresponding to a node $x\ell_i$ will be denoted by $X^\ell_i$, and its concentration at time $t$ will be denoted by $x^\ell_i(t)$.

\subsection{Backward CRN}\label{sec: backward network}
 The derivative $\sigma '$ is the key term to calculate the loss function gradient. When smooth ReLU is used, the values in the nodes are positive as $x^{\ell}=\sigma(z^{\ell-1})$. Hence, for the derivative of the smooth ReLU given in \eqref{eq:derivative of relu}, we employ the CRNs to calculate the addition and division of positive real numbers suggested in \cite{buisman2009computing} without introducing dual rails.

Let $X^\ell_i \in \Sp^f$ be the species corresponding to 
node $x^\ell_i$ in the feed-forward network. Using this as an input, the CRN for computing $\dfrac{\partial x^{\ell}_i}{\partial z^{\ell-1}_i}=\sigma'(z^{\ell-1}_i)$ is 
\begin{align}
\begin{split}\label{eq:crn for div and add}
    &2X^\ell_i\xrightarrow{\kappa_1} 2X^\ell_i+DX^\ell_i, \quad DX^\ell_i+A^\ell_i \xrightarrow{\kappa_2}  A^\ell_i\\
    &A^\ell_i\xrightarrow{\kappa_3}  \emptyset, \quad 2X^\ell_i\xrightarrow{\kappa_4}  2X^\ell_i+A^\ell_i, \quad H\xrightarrow{\kappa_5}  H+A^\ell_i.
    \end{split}
\end{align}
If we set rate constants as $\dfrac{\kappa_5 \kappa_2}{\kappa_3}=\dfrac{\kappa_4 \kappa_2}{\kappa_3}=\kappa_1$, then the concentration of $DX^\ell_i$, which represents the derivative $\dfrac{\partial x^\ell_i}{\partial z^{\ell-1}_i}=\sigma'(z^{\ell-1}_i)$, convergences to $\dfrac{(x^\ell_i)^2}{(x^\ell_i)^2+h}$. That is, 
$\dlim_{t\to \infty}dx^{\ell+1}_i(t)=\dlim_{t\to \infty}\dfrac{(x^\ell_i)^2(t)}{(x^\ell_i)^2+h}=\dfrac{(x^\ell_i)^2}{(x^\ell_i)^2+h}$. $A^\ell$ is a species computing the summation $(x^\ell_i)^2+h$, and species $H$ used in $(\Sp^f, \C^f, \Re^f,\K^f)$ corresponds to the smoothness parameter $h$ in \eqref{eq:smooth relu}. 

By multiplying and summing these derivative terms with $w^\ell_{ij}$ and $b^\ell_i$'s, the gradients of $\mathcal L$ can be computed recursively from $\ell=L$ to $\ell=1$ as described in Section \ref{sec: gradient descent}. Hence the full backward network consists of further sub-CRNs that compute multiplication and addition, and the $W^{\ell,\pm}_{ij}$ and $B^{\ell,\pm}_{i}$ species are used. It's important to note that the weights and biases can be negative, so we use the dual rails $W^{\ell,\pm}_{ij}$ and $B^{\ell,\pm}_{i}$ to ensure that the differences of two non-negative values $w^{\ell,+}_{ij}-w^{\ell,-}_{ij}$ and $b^{\ell,+}_i-b^{\ell,-}_i$ can represent negative values. The terms such as $\dfrac{\partial \mathcal L}{\partial x^\ell_i}$ and $\dfrac{\partial \mathcal L}{\partial w^\ell_{ij}}$ that are computed with $w^{\ell,\pm}_{ij}$ and $b^{\ell,\pm}_i$'s  are also realized with dual rails. For example, for a simple two-layered NN illustrated in Figure \ref{fig: structure}, to compute $\dfrac{\partial \mathcal L}{\partial x^1_1}=\dfrac{\partial \mathcal L}{\partial z^1_1}w^1_{11}$, we use
\begin{align*}
     &LX^{1,+}_1 \xrightarrow{\kappa_6} \emptyset, \quad LZ^1_1+W^{1,+}_{11}\xrightarrow{\kappa_7}  LZ^1_1+W^{1,+}_{11}+LX^{1,+}_1\\
    &LX^{1,-}_1 \xrightarrow{\kappa_8}  \emptyset, \quad LZ^1_1+W^{1,-}_{11}\xrightarrow{\kappa_9}  LZ^1_1+W^{1,-}_{11}+LX^{1,-}_1.
\end{align*}
 If $\kappa_i=1$ for $i=6,7,8,9$, then $\dlim_{t\to \infty}\left( lx^{1,+}_1(t)-lx^{1,-}_1(t)\right )=\dlim_{t\to \infty}lz^1_1(t)(w^{1,+}_1(t)-w^{1,-}_1(t))= \dfrac{\partial \mathcal L}{\partial z^1_1} w^1_1$. 
 
To show an example of the full construction, we provide the list of reactions in $(\Sp^b, \C^b, \Re^b,\K^b)$ and $(\Sp^u, \C^u, \Re^u,\K^u)$ for a simple two-layered NN in Section \ref{sec: full backward and update}.
We further visualize the architecture of $(\Sp^b, \C^b, \Re^b,\K^b)$ and $(\Sp^u, \C^u, \Re^u,\K^u)$ in Figure \ref{fig: structure}.

\begin{rem}\label{rem:feedback free}
  he arrows going into and out from a green box in Figure \ref{fig: structure} represent the catalyst-product relation. In other words, the species at the incoming arrows act as the catalyst species, while the species at the outgoing arrow are the dynamic species. These arrows in the backward CRN are always one-sided, and this also applies to the forward CRN \cite{anderson2021reaction}. Essentially, this means that the combination of the forward and backward networks has a 'feedback-free' structure. Specifically, if species $S_1$ acts as the catalyst to produce $S_2$, then $S_2$ never serves as a catalyst to produce $S_1
  $. This plays a crucial role in the convergence of the forward and backward networks, as well as in the noise analysis (See Section \ref{sec:one pot} and Section \ref{sec:noise analysis}).

   In contrast, the update network gives a feedback loop to the forward network. We will slow down these feedback reactions to prevent disruption of the forward and backward CRNs (See Section \ref{sec:update}).
\end{rem}


\begin{figure}[!h]  \centerline{\includegraphics[width=1.1\textwidth]{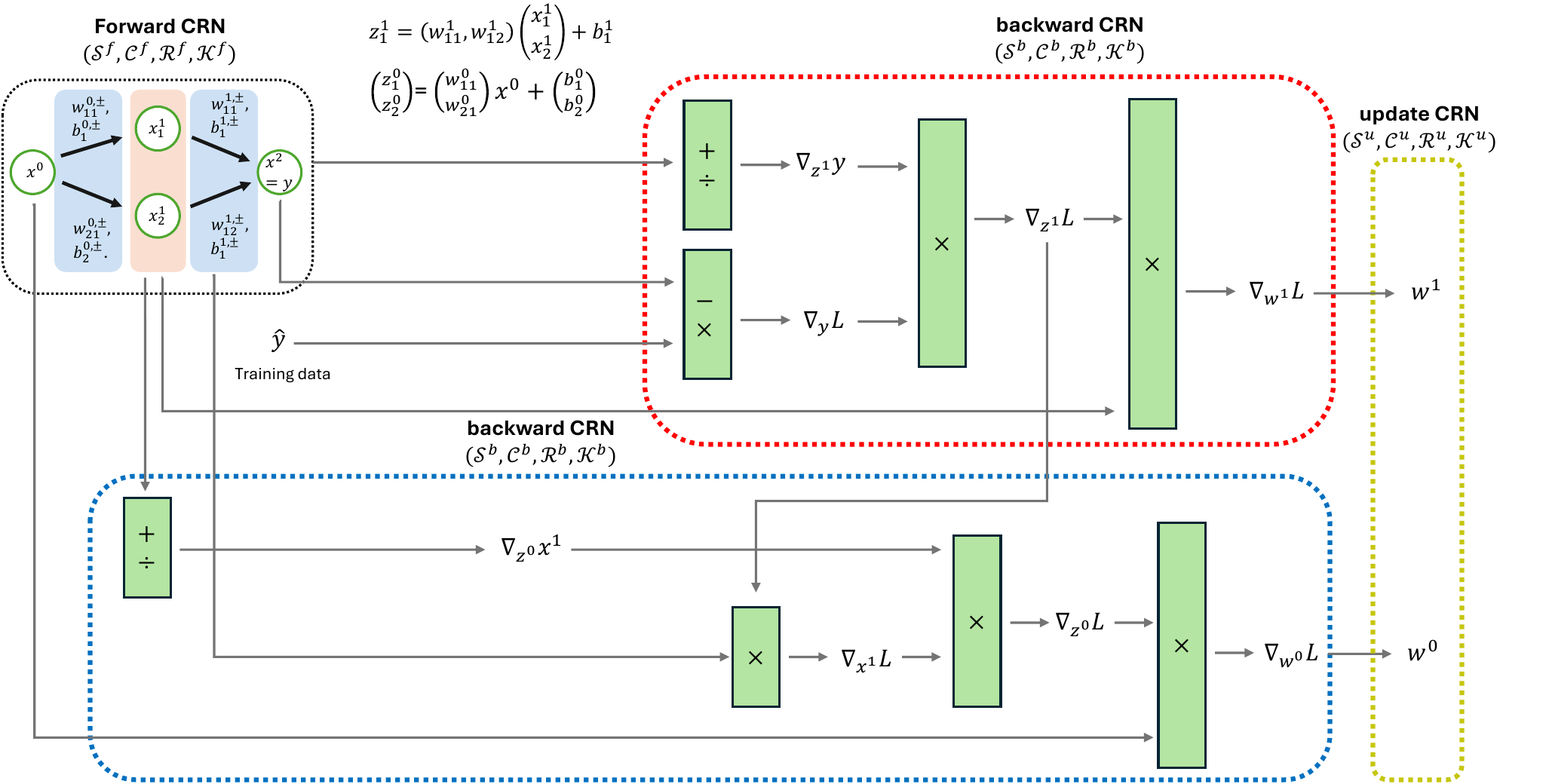}}
    \caption{Overall structure of the CRN implementing the backpropagation of a two-layered NN. The green boxes indicate a CRN that calculates the designated arithmetic operations. Each green box has two input species indicated by the incoming arrows.  The outgoing arrows indicate the output species. The full lists of reactions for the backward and the update CRNs are provided in Section \ref{sec: full backward and update}.}\label{fig: structure}
\end{figure}


\subsection{Update network}\label{sec:update}
We continuously update the weights using the gradient of the loss function, rather than using the usual iterative update.
Let $W^{\ell,+}_{ij}$ and $W^{\ell,-}_{ij}$ be species corresponding to the weight $w^\ell_{ij}=w^{\ell,+}_{ij}-w^{\ell,-}_{ij}$  in the forward CRN $(\Sp^f,\C^f,\Re^f,\K^f)$.   Then we build the update network using the following reactions 
\begin{align}\label{eq:update network}
    LW^{\ell,-}_{ij} \xrightarrow{\eta} LW^{\ell,-}_{ij} + W^{\ell,+}_{ij}, \quad  LW^{\ell,+}_{ij} \xrightarrow{\eta}  LW^{\ell,+}_{ij} + W^{\ell,-}_{ij},
\end{align}
where $LW^{\ell,\pm}_{ij}$ are the species that compute the gradient in the backward network $(\Sp^b,\C^b,\Re^b,\K^b)$ as
\begin{align*}
    \lim_{t\to \infty}\left (lw^{\ell,+}_{ij}(t)-lw^{\ell,-}_{ij}(t) \right )=\frac{\partial \mathcal L}{\partial w^\ell_{ij}}.
\end{align*}
We have that from \eqref{eq:update network}
\begin{align}\label{eq:update}
    \frac{d}{dt}(w^{\ell,+}_{ij}-w^{\ell,+}_{ij}) = -\eta \left (lw^{\ell,+}_{ij}(t)-lw^{\ell,-}_{ij}(t) \right ),
\end{align}
  leading to the continuous update of $w^{\ell}_{ij}$ with the drift given by $lw^{\ell,+}_{ij}$ and $lw^{\ell,-}_{ij}$. This continuous update indeed mimics the usual iterative update in the following sense. The continuous update 
 \eqref{eq:update} yields that 
\begin{align}\label{eq:solve w}
   w^\ell_{ij}(T)= w^{\ell,+}_{ij}(T)-w^{\ell,-}_{ij}(T) = \left( w^{\ell,+}_{ij}(0)-w^{\ell,-}_{ij}(0) \right)-\eta \int_0^{T} (lw^{\ell,+}_{ij}(s)-lw^{\ell,-}_{ij}(s))ds,
    \end{align}
     In Appendix \ref{app:conv speed}, we show that the species in the backward CRN including  $LW^{\ell,\pm}_{ij}$'s converge exponentially quickly to the steady state. Therefore for a large $T$ and for a sufficiently small $\eta$, \eqref{eq:solve w} can be approximated as
    \begin{align}
         w^\ell_{ij}(T)= w^{\ell,+}_{ij}(T)-w^{\ell,-}_{ij}(T) &\approx \left( w^{\ell,+}_{ij}(0)-w^{\ell,-}_{ij}(0) \right)-\eta T\lim_{t\to \infty}(lw^{\ell,+}_{ij}(T)-lw^{\ell,-}_{ij}(T)) \notag \\
         &=\left( w^{\ell,+}_{ij}(0)-w^{\ell,-}_{ij}(0) \right)-\eta T\frac{\partial \mathcal L}{\partial w^\ell_{ij}}. \label{eq:app iterative}
    \end{align}
\eqref{eq:app iterative} is equivalent to updating the parameters using the iterative update in the gradient descent scheme with the learning rate $\eta T$.

Similarly, we can update species $B^{\ell,+}_{ij}$ and $B^{\ell,-}_{ij}$ that represent the bias parameter $b^\ell_{ij}$ using a species, namely $LZ^\ell_i$, corresponding to $\frac{\partial \mathcal L}{\partial z^\ell_i}=\frac{\partial \mathcal L}{\partial b^\ell_i}$. In Section \ref{sec:one pot}, the convergence of the full CRN and the update of the parameters in the manner of a one-pot computation are more precisely described.

\subsection{One-pot computing and convergence}\label{sec:one pot}
The proposed CRN implements both the feed-forward and the training processes of a neural network in a single pot. This means that all reactions occur simultaneously without engineered circuits or oscillations that are used to run the forward and trainning parts separately. The previously proposed CRNs for trainning NNs use oscillatory clock-proteins to sequentially regulate reactions \cite{banda2013online, lakin2023design} for interactive updates of the weights and biases of NNs. 
The main task of performing one-pot computing was to design the update network introduced in Section \ref{sec:update}. This is because as the species in the update CRN gradually change, the convergence of the species in the forward and backward CRNs holds even when their reactions take place simultaneously. In this section, we will validate this. 
\begin{figure}[!h]
    \centering
\includegraphics[width=1.05\linewidth]{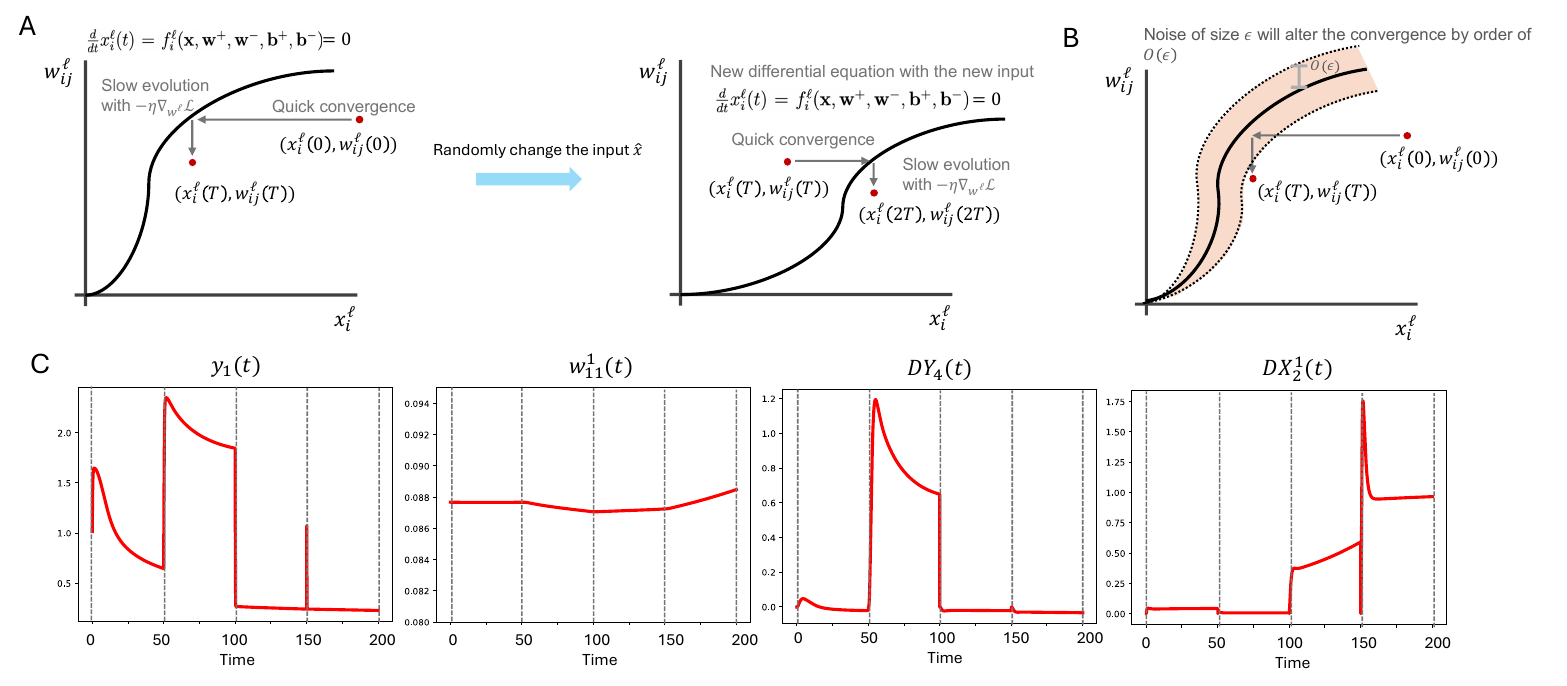}
    \caption{A and B. A schematic description of the convergence of the backward and the update networks without noise (A)  and with noise (B). C. The time evolutions of the species. The vertical dotted lines indicate the times when the input of the NN is changed randomly. In particular, the concentration $w^1_{11}(t)$ of the species $W^1_{11}$ corresponding to a weight parameter evolves slowly. }
    \label{fig:update and time evolution}
\end{figure}

Although it is difficult to thoroughly validate the convergence of the one-pot computation, we can show the convergence of the leading-order term of the species using a perturbation analysis regarding $\eta$ in \eqref{eq:update network} as a scaling parameter. We consider the system of differential equations for $x^\ell_i, w^{\ell,\pm}_{ij}$ and $b^{\ell,\pm}_i$'s as a perturbed system by the updated network with the perturbation parameter $\eta$.

To consider perturbed solutions, for $t$ on a compact time interval $[0,T]$, we decompose the concentrations of the species as
\begin{align*}
    &x^\ell_i(t)=x^\ell_{i,0}(t)+\eta x^\ell_{i,1}(t) + O(\eta^2),\\
&w^{\ell,\pm}_{ij}(t)=w^{\ell,\pm}_{ij,0}(t)+\eta w^{\ell,\pm}_{ij,1}(t) + O(\eta^2), \text{ and}\\
&b^{\ell,\pm}_{i}(t)=b^{\ell,\pm}_{i,0}(t)+\eta b^{\ell,\pm}_{i,1}(t) + O(\eta^2).
\end{align*} 
Note that the concentration $x^\ell_i(t)$ of the species corresponding to a node follows a differential equation of the form $\frac{d}{dt}x^\ell_i(t)=f^\ell_i(\mathbf x, \mathbf w^+, \mathbf w^-, \mathbf b^+,\mathbf b^-)$ for some function $f^\ell_i$, where $\mathbf x$ and $\mathbf b^{\pm}$ denote matrices whose $\ell i$ components are $x^\ell_{i}$ and $b^\ell_{i}$, and the $\ell i j$ component of $\mathbf w^\pm$ is $w^{\ell,\pm}_{ij}$, respectively. We similarly denote by $\mathbf x_k$ and $\mathbf b^{\pm}_k$  the matrices whose $\ell i$ components are $x^\ell_{i,0}$ and $b^\ell_{i,k}$, and we let the $\ell i j$ component of $\mathbf{w^\pm}_0$ be $w^{\ell,\pm}_{ij,k}$, respectively, for $k=1, 2$.
Importantly, due to smoothness, $f^\ell_i$ can also be decomposed with the Talyor expansion as 
\begin{align*}
   \frac{d}{dt}x^\ell_i(t)&=f^\ell_i(\mathbf x_0, \mathbf w^+_0, \mathbf w^-_0, \mathbf b^+_0, \mathbf b^-_0)+
    \nabla f^\ell_i((\mathbf x_0, \mathbf w^+_0, \mathbf w^-_0, \mathbf b^+_0, \mathbf b^-_0)) \cdot \Delta + O(\eta^2)\\
    &=f^\ell_i(\mathbf x_0, \mathbf w^+_0, \mathbf w^-_0, \mathbf b^+_0, \mathbf b^-_0)+
    \eta \nabla f^\ell_i((\mathbf x_0, \mathbf w^+_0, \mathbf w^-_0, \mathbf b^+_0, \mathbf b^-_0) \cdot  (\mathbf x_1, \mathbf w^+_1, \mathbf w^-_1, \mathbf b^+_1, \mathbf b^-_1)^\top+ O(\eta^2),
\end{align*}
where $\Delta=(\mathbf x, \mathbf w^+, \mathbf w^-, \mathbf b^+, \mathbf b^-)^\top- (\mathbf x_0, \mathbf w^+_0, \mathbf w^-_0, \mathbf b^+_0, \mathbf b^-_0)^\top=\eta(\mathbf x_1, \mathbf w^+_1, \mathbf w^-_1, \mathbf b^+_1, \mathbf b^-_1)^\top + O(\eta^2)$.
Note that the differential equations for species $W^{\ell,\pm}_{ij}$ and $B^{\ell,\pm}_i$ are evolving with catalysts $LW^{\ell,\pm}_{ij}$ and $LB^{\ell,\pm}_i$ as shown in \eqref{eq:update network}, where $LW^{\ell,\pm}_{ij}$ and $LB^{\ell,\pm}_i$ correspond to the derivatives $\dfrac{\partial \mathcal L}{\partial w^\ell_{ij}}$ and $\dfrac{\partial \mathcal L}{\partial b^\ell_{i}}$. These are represented as summation, product, and division (with strictly positive denominators) of $x^\ell_i, w^{\ell,\pm}_{ij}$ and $b^{\ell,\pm}_i$ by the backward CRN (Section \ref{sec: backward network}). Hence we can view $w^{\ell,\pm}_{ij}$ and $b^{\ell,\pm}_i$ as the solutions of 
\begin{align*}
     &\frac{d}{dt}w^{\ell,\pm}_{ij}(t)=-\eta \times lw^{\ell,\mp}_{ij}(t):=\eta g^{\ell,\pm}_{ij}(\mathbf x, \mathbf w^+, \mathbf w^-, \mathbf b^+, \mathbf b^-)\\
     &\frac{d}{dt}b^{\ell,\pm}_{i}(t)=-\eta \times lz^{\ell,\mp}_{i}(t):=\eta q^{\ell,\pm}_{i}(\mathbf x, \mathbf w^+, \mathbf w^-, \mathbf b^+, \mathbf b^-),
\end{align*}
where the functions $g^{\ell,\pm}_{ij}$ and $q^{\ell,\pm}_{i}$ are smooth with respect to $\mathbf x, \mathbf w^+, \mathbf w^-, \mathbf b^+$, and $\mathbf b^-$.
Then by the Taylor expansion we get that
\begin{align}
\begin{split}\label{eq:wb orders}
  \frac{d}{dt}w^{\ell,\pm}_{ij}(t) &=\eta\left (g^{\ell,\pm}_{ij}(\mathbf x_0, \mathbf w^+_0, \mathbf w^-_0, \mathbf b^+_0, \mathbf b^-_0)^\top+\eta \nabla g^{\ell,\pm}_{ij}\cdot (\mathbf x_1, \mathbf w^+_1, \mathbf w^-_1, \mathbf b^+_1, \mathbf b^-_1)^\top +O(\eta^2) \right )\\
  \frac{d}{dt}b^{\ell,\pm}_{i}(t) &=\eta\left (q^{\ell,\pm}_{i}(\mathbf x_0, \mathbf w^+_0, \mathbf w^-_0, \mathbf b^+_0, \mathbf b^-_0)^\top+\eta \nabla q^{\ell,\pm}_{i}\cdot (\mathbf x_1, \mathbf w^+_1, \mathbf w^-_1, \mathbf b^+_1, \mathbf b^-_1)^\top +O(\eta^2) \right )
      \end{split}
\end{align}
By matching the orders with respect to $\eta$, we have
\begin{align*}
    \frac{d}{dt}x^\ell_{i,0}(t)&=f^\ell_i(\mathbf x_0, \mathbf w^+_0, \mathbf w^-_0, \mathbf b^+_0, \mathbf b^-_0),\\
     \frac{d}{dt}w^{\ell,\pm}_{ij,0}(t)&=0, \text{ and}\\
      \frac{d}{dt}b^{\ell,\pm}_{i,0}(t)&=0,
\end{align*}
and hence the leading order term $x^\ell_{i,0}(t)$ follows the dynamics governed by $f^\ell_i$ while the leading order term $w^{\ell,\pm}_{ij,0}$ and $b^{\ell,\pm}_{i,0}$ are fixed at the initial condition. Furthermore, the convergence speed of the system $\frac{d}{dt}x^\ell_{i,0}(t)=f^\ell_i(\mathbf x_0, \mathbf w^+_0, \mathbf w^-_0, \mathbf b^+_0, \mathbf b^-_0)$ is exponential in time (Appendix \ref{app:conv speed}). Therefore $x^\ell_{i,0}(T)$ will closely approximate its steady state for a large $T$. After this convergence is approximately made, the dynamics of the system can be demonstrated with the evolution of the second leading terms of order $\eta$. Matching the order in \eqref{eq:wb orders}, we find that $w^{\ell,\pm}_{ij,1}$ and $b^{\ell,\pm}_{i,1}$ follow 
\begin{align*}
    &\frac{d}{dt}w^{\ell,\pm}_{ij,1}= g^{\ell,\mp}_{ij}(\mathbf x_0, \mathbf w^+_0, \mathbf w^-_0, \mathbf b^+_0, \mathbf b^-_0)\approx \frac{\partial \mathcal L}{\partial w^{\ell,\pm}_{ij}}(\mathbf x_0(\infty),\mathbf w^+_0(0), \mathbf w^-_0(0), \mathbf b^+_0(0), \mathbf b^-_0(0)),\\
    &\frac{d}{dt}b^{\ell,\pm}_{i,1}= q^{\ell,\mp}_{i}(\mathbf x_0, \mathbf w^+_0, \mathbf w^-_0, \mathbf b^+_0, \mathbf b^-_0)\approx \frac{\partial \mathcal L}{\partial b^{\ell,\pm}_i} (\mathbf x_0(\infty),\mathbf w^+_0(0), \mathbf w^-_0(0), \mathbf b^+_0(0), \mathbf b^-_0(0)),
\end{align*}
where $\mathbf x_0(\infty):=\lim_{t\to \infty}\mathbf x_0$. In this approximation, we used the exponentially quick convergence of $x^\ell_i$'s. 
Hence, we get the desired gradient descent as
\begin{align} \begin{split}\label{eq:update of w and b}
w^{\ell,\pm}_{ij}(T)&\approx w^{\ell,\pm}_{ij,0}(0)+\eta  T \frac{\partial \mathcal L}{\partial w^{\ell,\pm}_{ij}}(\mathbf x_0(\infty),\mathbf w^+_0(0), \mathbf w^-_0(0), \mathbf b^+_0(0), \mathbf b^-_0(0)),\\
      b^{\ell,\pm}_{i}&\approx b^{\ell,\pm}_{i,0}(0)+\eta  T \frac{\partial \mathcal L}{\partial b^{\ell,\pm}_i} (\mathbf x_0(\infty),\mathbf w^+_0(0), \mathbf w^-_0(0), \mathbf b^+_0(0), \mathbf b^-_0(0)).
      \end{split}
\end{align}
 This process is repeated with random changes of the inputs $\hat x$.  The convergence of the backward and the update network is schematically described in Figure \ref{fig:update and time evolution} A. We also provided schematic the convergence when noise of reaction rate constants exists in Figure \ref{fig:update and time evolution} B.
 The time evolutions of species in our CRN implementing a NN are displayed in Figures \ref{fig:update and time evolution} C.

    \subsection{An example of the the backward and the update network} \label{sec: full backward and update}

Here we provide a complete list of reactions in $(\Sp^b,\C^b,\Re^b,\K^b)$ and $(\Sp^u,\C^u,\Re^u,\K^u)$ to implement the backward network of a simple two-layered NN described in Figure \ref{fig: structure}.

In Table \ref{table:1}, w presents the CRN that calculates the derivative of $\mathcal L$ with respect to the output layer parameters. The concentration of a species is denoted by lowercase letters. Unless specified, the rate parameters are assumed to be $1$ and are omitted. 
\begin{table}[!h]
\begin{tabular}{|c|c|}
\hline
sub-CRNs to compute $\nabla_{w^1}\mathcal L$ and $\nabla_{b^1}\mathcal L$& Goal \\
\hline
\makecell[l]{
$\begin{aligned}
    &2Y\to 2Y+DY, \quad DY+A^2\to A^2\\
    &A^2\to \emptyset, \quad 2Y\to 2Y+A^2, \quad H\to H+A^2
\end{aligned}$} 
&$\dlim_{t\to \infty}dy(t)= \frac{y^2}{y^2+h}=\frac{\partial y}{\partial z^1_1}$ \\
\hline
$\begin{aligned}
    Y\xrightarrow{2}LY+Y, \quad \hat Y\xrightarrow{2}LY+\hat Y,\quad LY\to \emptyset
\end{aligned}$ 
& $\dlim_{t\to \infty}ly(t)=2(y-\hat y)=\frac{\partial \mathcal L}{\partial y}$ \\
\hline
$\begin{aligned}
    LZ^1_1\to \emptyset, \quad LY+DY\to LY+DY+LZ^0_1
\end{aligned}$ 
& $\dlim_{t\to \infty}lz^1_1(t)=2(y-\hat y)\frac{\partial y}{\partial z^1_1}=\frac{\partial  \mathcal L}{\partial z^1_1}$\\
\hline
$\begin{aligned}
   &LW^1_{11}\to \emptyset, \quad LZ^1_1+X^1_1\to LZ^1_1+X^1_1+LW^0_{11}\\
   &LW^1_{12}\to \emptyset, \quad LZ^1_1+X^1_2\to LZ^1_1+X^1_2+LW^0_{12}
\end{aligned}$ 
& $\begin{aligned}
    &\dlim_{t\to \infty}lw^1_{11}(t)=\frac{\partial \mathcal L}{\partial z^1_1}\frac{\partial z^1_1}{\partial w^1_{11}}=\frac{\partial \mathcal L}{\partial w^1_{11}}\\
     &\dlim_{t\to \infty}lw^1_{12}(t)=\frac{\partial \mathcal L}{\partial z^1_1}\frac{\partial z^1_1}{\partial w^1_{12}}=\frac{\partial \mathcal L}{\partial w^1_{12}}
\end{aligned}$\\
\hline
\end{tabular}
\caption{An example of a sub-CRN in $(\Sp^b,\C^b,\Re^b,\K^b)$ that computes $\nabla_{w^1}\mathcal L$ and $\nabla_{b^1}\mathcal L$.}\label{table:1}
\end{table}
Note that we do not need to compute $\frac{\partial \mathcal L}{\partial b^1_1}$ separately as it is equal to $\frac{\partial L}{\partial z^1_1}$.
The CRN that computes the derivative of $\mathcal L$ with respect to the parameters of the first layer is shown in Table \ref{table:2}. 
\begin{table}[!h]
\begin{tabular}{|c|c|}
\hline
sub-CRNs to compute $\nabla_{w^0}\mathcal L$ and $\nabla_{b^0}\mathcal L$ & Goal \\
\hline
$\begin{aligned}
    &2X^1_1\to 2X^1_1+DX^1_1, \quad DX^1_1+A^1_1\to A^1_1\\
    &A^1_1\to \emptyset, \quad 2X^1_1\to 2X^1_1+A^1_1, \quad H\to H+A^1_1
\end{aligned}$
&$\dlim_{t\to \infty}dx^1_1(t)= \frac{(x^1_1)^2}{(x^1_1)^2+h}=\frac{\partial x^1_1}{\partial z^0_1}$ \\
\hline
$\begin{aligned}
   LX^{1,\pm}_1 \to \emptyset, \quad LZ^1_1+W^{1,\pm}_{11}\to LZ^1_1+W^{1,\pm}_{11}+LX^{1,\pm}_1
\end{aligned}$ 
& $\dlim_{t\to \infty}(lx^{1,+}_1(t)-lx^{1,-}_1(t))=\frac{\partial \mathcal L}{\partial x^1_1}$ \\
\hline
$\begin{aligned}
    LZ^{0,\pm}_1\to \emptyset, \quad LX^{1,\pm}_1+DX^1_1\to LX^{1,\pm}_1+DX^1_1+LZ^{0,\pm}_1
\end{aligned}$ 
& $\dlim_{t\to \infty}(lz^{0,+}_1(t)-lz^{0,-}_1(t))=\frac{\partial \mathcal L}{\partial z^0_1}$\\
\hline
$\begin{aligned}
   LW^{0,\pm}_{11}\to \emptyset, \quad  LZ^{0,\pm}_1+X^0\to  LZ^{0,\pm}_1+X^0+ LW^{0,\pm}_{11}
\end{aligned}$ 
& $\dlim_{t\to \infty}(lw^{0,+}_{11}(t)-lw^{0,-}_{11}(t))=\frac{\partial \mathcal L}{\partial w^0_{11}}$\\
\hline
$\begin{aligned}
    &2X^1_2\to 2X^1_2+DX^1_2, \quad DX^1_2+A^1_2\to A^1_2\\
    &A^1_2\to \emptyset, \quad 2X^1_2\to 2X^1_2+A^1_2, \quad H\to H+A^1_2
\end{aligned}$
&$\dlim_{t\to \infty}dx^1_2(t)=\frac{(x^1_2)^2}{(x^1_2)^2+h}=\frac{\partial x^1_2}{\partial z^0_2}$ \\
\hline
$\begin{aligned}
   LX^{1,\pm}_2 \to \emptyset, \quad LZ^1_1+W^{1,\pm}_{12}\to LZ^1_1+W^{1,\pm}_{12}+LX^{1,\pm}_2
\end{aligned}$ 
& $\dlim_{t\to \infty}(lx^{1,+}_2(t)-lx^{1,-}_2(t))=\frac{\partial \mathcal L}{\partial x^1_2}$ \\
\hline
$\begin{aligned}
    LZ^{0,\pm}_2\to \emptyset, \quad LX^{1,\pm}_2+DX^1_2\to LX^{1,\pm}_2+DX^1_2+LZ^{0,\pm}_2
\end{aligned}$ 
& $\dlim_{t\to \infty}(lz^{0,+}_2(t)-lz^{0,-}_2(t))=\frac{\partial \mathcal L}{\partial z^0_2}$\\
\hline
$\begin{aligned}
   LW^{0,\pm}_{21}\to \emptyset, \quad  LZ^{0,\pm}_2+X^0\to  LZ^{0,\pm}_2+X^0+ LW^{0,\pm}_{21}
\end{aligned}$ 
& $\dlim_{t\to \infty}(lw^{0,+}_{21}(t)-lw^{0,-}_{21}(t))=\frac{\partial \mathcal L}{\partial w^0_{21}}$\\
\hline
\end{tabular}
\caption{An example of a sub-CRN in $(\Sp^b,\C^b,\Re^b,\K^b)$ that computes $\nabla_{w^0}\mathcal L$ and $\nabla_{b^0}\mathcal L$.}\label{table:2}
\end{table}
Similarly, we do not need to compute $\frac{\partial \mathcal L}{\partial b^{0}_i}$ separately as it is equal to $\frac{\partial \mathcal L}{\partial z^{0}_i}$ for each $i=1,2$.
The overall architecture of the backward network is illustrated in Figure \ref{fig: structure}. 
Finally, we provide reactions in the update CRN in Table \ref{table:3}.
\begin{table}
\begin{tabular}{|c|c|}
\hline
Update network & Goal\\
\hline
\makecell[l]{
$\begin{aligned}
   \hspace{3cm} LW^{1,\mp}_{11} \xrightarrow{\eta} LW^{1,\mp}_{11} + W^{1,\pm}_{11}  \hspace{3cm}
\end{aligned}$} 
&$w^1_{11} \mapsfrom w^1_{11}-\eta T\frac{\partial \mathcal L}{\partial w^1_{11}}$ \\
\hline
$\begin{aligned}
   LW^{1,\mp}_{12} \xrightarrow{\eta} LW^{1,\mp}_{12} + W^{1,\pm}_{12}
\end{aligned}$
& $w^1_{12} \mapsfrom w^1_{12}-\eta T\frac{\partial \mathcal L}{\partial w^1_{12}}$\\
\hline
$\begin{aligned}
    LW^{0,\mp}_{11} \xrightarrow{\eta} LW^{0,\mp}_{11} + W^{0,\pm}_{11}
\end{aligned}$ 
& $w^0_{11} \mapsfrom w^0_{11}-\eta T\frac{\partial \mathcal L}{\partial w^0_{11}}$\\
\hline
$\begin{aligned}
    LW^{0,\mp}_{21} \xrightarrow{\eta} LW^{0,\mp}_{21} + W^{0,\pm}_{21}
\end{aligned}$ 
& $w^0_{21} \mapsfrom w^0_{21}-\eta T\frac{\partial \mathcal L}{\partial w^0_{21}}$\\
\hline
\end{tabular}
\caption{An example of $(\Sp^u,\C^u,\Re^u,\K^u)$ }\label{table:3}
\end{table}
The construction of the forward network follows the previous work \cite{anderson2021reaction}.

\section{Applications to data sets}\label{sec:Xor, Iris}
In this section, we applied the complete CRN, which includes the forward, backward, and update CRNs, to implement neural networks on various datasets. The forward network used is the same as in \cite{anderson2021reaction}, while the backward and update networks are constructed as proposed in Section \ref{sec: backward network} and Section \ref{sec:update}. 
The datasets used include XOR, Iris, and MNIST, which are typical examples of classification problems in machine learning. More detailed descriptions of the datasets are provided in Appendix \ref{app:data}.
For these datasets, we utilized a two-layered neural network consisting of the input layer, a single hidden layer, and the output layer, implemented with the proposed CRNs. We regard running our CRN over an interval $[(k-1)T,kT)$ for each $k$ as a single iteration. After each iteration, we randomly changed the input $\hat x$. We let the CRN run for $[0,KT)$, where $K$ means the total number of iterations.

We used an error rate function to measure the performance of our CRN. Here, the error rate $E(t)$ is the average of the mean squared error (MSE) across the entire dataset at time $t$. More precisely, 
\begin{align}\label{eq:error}
    E(t)= \frac{1}{N}\sum_{n=1}^N \Vert y_n(t)-\hat y_n\Vert_{L_2}^2,
\end{align}
where $y_n(t)=(y_{n,1}(t),\dots,y_{n,M^L}(t))$ is the vector of the concentration of the species corresponding to the output layer with $n$ th input data. Importantly when evaluating $y_n(t)$, we used the weights $w^\ell_{ij}(t)$'s and biases $b^\ell_{i}(t)$ to measure the evolution of the error. $\hat y_n$ is the training data for the $n$ th input data.
Here, $N$ is the number of total data sets.
That is, $E(t)$ measures the average error of the neural network at time $t$ using the weights and biases that have been updated through the backward and the updated CRNs up to time $t$.

With four XOR data points, we trained the forward CRN using the backward and update CRNs proposed in Section \ref{sec: backward network} and \ref{sec:update}. 
We plotted the time evolution of the error and the accuracy of the NN implemented by the CRN (Figure \ref{fig:application to data} A left).  
The graphs of the fitted functions are displayed in  Figure \ref{fig:application to data} A right.
For the Iris data sets, we plotted the time evolutions of the error functions and the accuracy (Figure \ref{fig:application to data} B top) and also provided the pair plots (Figure \ref{fig:application to data} B bottom).

For MNIST data sets, each input data point is a $8\times 8$ pixel representing an image of a handwritten digit. In addition to the error and accuracy functions ((Figure \ref{fig:application to data} C top), we also displayed the classification of the input data points that are embedded in a plane through principal component analysis (PCA) (Figure \ref{fig:application to data} C bottom). The classification by our CRN is nearly the same as the classification by a traditional NN. Note that for the Iris and MNIST data sets, we splitted the original data sets into training and validation data sets. When $\hat y_n$ comes from the training (validation) data sets, we call $E(t)$ the training (validation) error.

For each simulation, all reaction rate constants within the forward and backward CRNs were assigned a value of $1$, with the exception of those related to division, such as $\kappa_1$ and $\kappa_2$ in \eqref{eq:crn for div and add}. These division-related rate constants were deliberately set to be large to enhance the speed of the reactions, as they tend to occur at a slow rate. Conversely, a small reaction rate constant $\eta$ was utilized for the update CRN.

See Appendix \ref{app:data} for the details of these experiments and the data sets.
\begin{figure}
    \centering \includegraphics[width=1\linewidth]{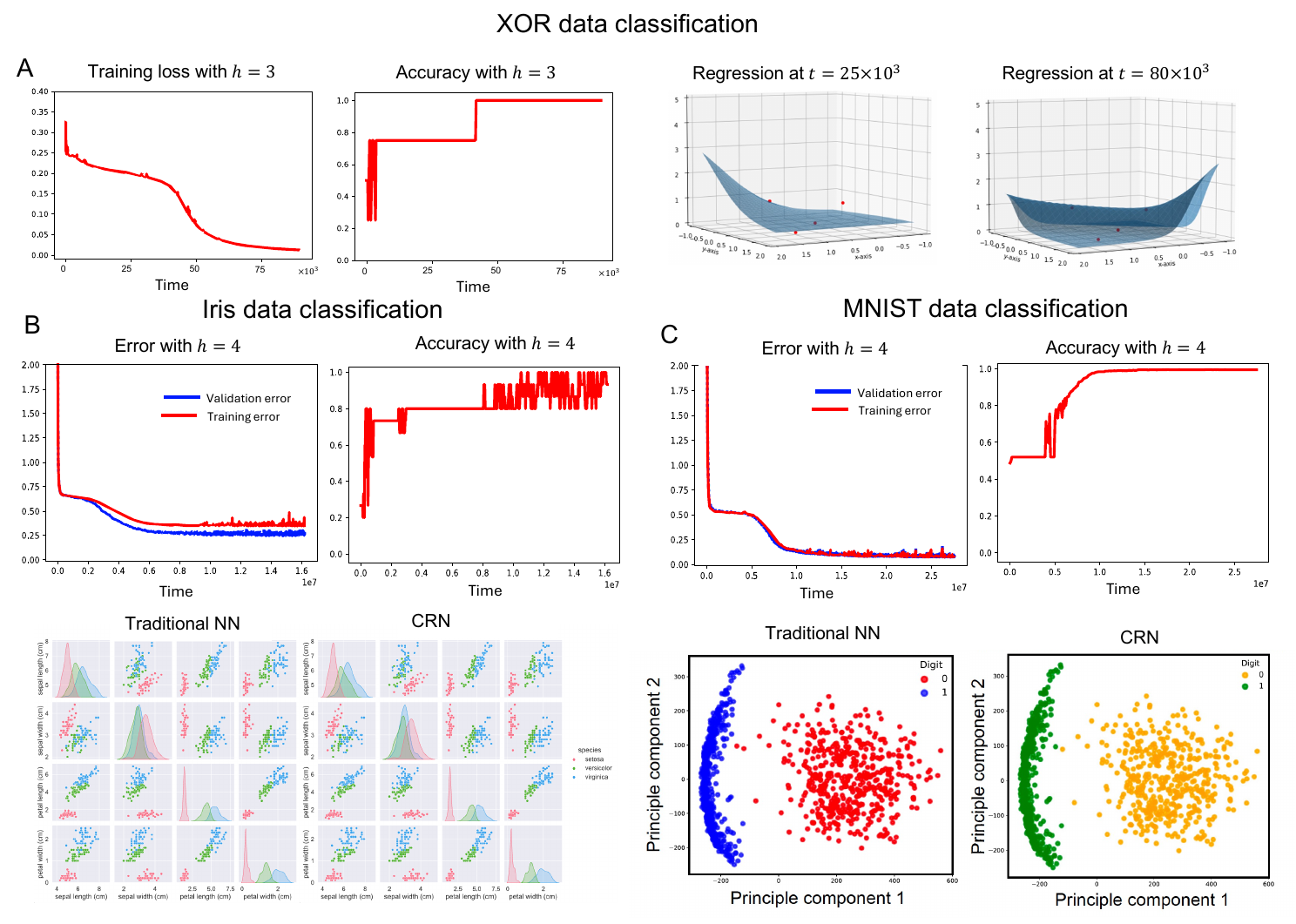}
    \caption{The time evolutions of the error and the accuracy functions for the XOR (A left), the Iris (B top), and the MNIST (C top) data sets. A right. The regression of the XOR data sets after 100 and 10000 iterations. B bottom. The pair plots for the Iris data sets (See Appnedix \ref{app:iris} C bottom. Classification of the MNIST input data sets by the outputs obtained with a standard NN and our CRN. 
    The input data points are embedded on a 2d plane via the principal component analysis. See Table \ref{table:parameters} in Supplementary material for the parameters used for these simulations. }
    \label{fig:application to data}
\end{figure}


\section{Robustness to noisy reaction rates}\label{sec:robust}
In reality, the reaction timing is affected by the fluctuation of temperature, pressure, and other external disturbances. This means that the rates of the reactions are inherently random. Due to intrinsic noise of the system, the concentrations of species can also fluctuate around the exact values modeled by the differential equations.   

These random fluctuations in either the reaction rate parameters or the concentration of the species will subsequently cause noise in the gradient of the loss function and eventually make the gradient descent process noisy. Recent studies have examined the effect of noise in stochastic gradient descent processes. Surprisingly, noise can regularize the optimization process and help find flatter local minima  \cite{zhu2018anisotropic}. However, a large variance of noise in the gradient descent processes can significantly distort optimization and slow down the convergence speed. Thus noise reduction methods have been developed in previous studies \cite{huang2021stochastic} where the variance of Gaussian noise is set to be a decreasing function of either time or the number of parameters. 
To robustly implement NNs against noisy parameters and species via the CRNs, therefore, it is critical to minimize fluctuations in the gradient descent process.

In this section, we proved that the smoothness of the activation function ensures that the size of the fluctuation in our CRN remains at most of order $\epsilon$ when both the extrinsic and the intrinsic noise are of order $\epsilon$. 

\subsection{Noisy gradient descent processes induced by fluctuations in the CRN}\label{sec:noise analysis}
We will begin by mathematically formulating the impact of noise from both the forward and backward CRNs on the update CRN. Since the input of the neural network is updated at $t=kT$ for each $k=1,2,\dots$ for some $T$, we will analyze the noise effect on $t\in[0,T]$.

We define the noisy rate parameters as follows. Let $\kappa_j(t;\epsilon)$ represent the noisy rate parameter of the $j$ th reaction. We assume that $\kappa_j(t;\epsilon)=|\kappa_j +\epsilon \xi_j(t)|$, where $\epsilon$ is the noise magnitude. The noise terms $\xi_j(t)$ follow independent standard normal distributions $\mathcal N(0,1^2)$. For some $t_0$ and $t_1$, they are re-sampled at each $t=k t_0$ and remain constant for $t_1$ for each $k=1,2,\dots$ (see Figure \ref{fig:activation_noise} A). That is, for each $j$,
\begin{align}\label{eq:normal}
\xi_j(t)=\begin{cases}Z_{j,k} \overset{\text{i.i.d.}}{\sim} \mathcal N(0,1^2) \quad &\text{if $t\in [kt_0, k t_0+t_1)$ for some $k$.}\\
    0 &\text{otherwise}, \end{cases}
\end{align}
where $\kappa_j$'s are the exact noise-free rate parameters. As mentioned before we use identical rate parameters to compute desired outcomes so that we set $\kappa_j=1$ for each $j$.
Let $(\Sp,\C,\Re, \mathcal K^\epsilon)$ represent the union of the forward CRN (employed from \cite{anderson2021reaction}) and the backward CRN that implements a NN with the smooth ReLU activation function \eqref{eq:smooth relu}. Here $\mathcal K^\epsilon$ is the set of the noisy rate parameters $\kappa_j(t;\epsilon)$'s. Let $(\Sp,\C,\Re,\mathcal K)$ be the same CRN as $(\Sp,\C,\Re, \mathcal K^\epsilon)$ but with noise-free parameters, where the reaction rate parameters are all equal to $1$ in $\mathcal K$.

Now we investigate the variation of the concentrations of the species by the noisy rate parameters.
Let $x^\ell_i(t;\epsilon)$ and $x^\ell_i(t)$ be the concentration of species $X^\ell_i$ of $(\Sp,\C,\Re,\mathcal K^\epsilon)$ and $(\Sp,\C,\Re,\mathcal K)$, respectively, that are corresponding to the $i$ th nodes on the $\ell$ th layer. Due to the structure of the forward network proposed in \cite{anderson2021reaction}, the governing equation of $x^\ell_j(t;\epsilon)$ is written as
\begin{align}
    \frac{d}{dt}x^\ell_i(t;\epsilon)&=h+x^\ell_i(t;\epsilon) \sum_{j}\rho^{\ell-1}_{j,\epsilon}(x^{\ell-1}_j(t;\epsilon))-(1+\epsilon \xi_k(t))(x^\ell_i(t;\epsilon))^2, \label{eq:forward eq}\\
     \frac{d}{dt}x^\ell_i(t)&=h+x^\ell_i(t) \sum_{j}\rho^{\ell-1}_{j,0}(x^{\ell-1}_j(t))-(x^\ell_i(t))^2, \label{eq:forward eq2}
\end{align}
where
\begin{align*}
\rho^{\ell}_{i,\epsilon}(z)=\sum_{j=1}^{M^\ell}\left((1+\epsilon \xi^{\ell,+}_{ij})w^{\ell,+}_{ij} -(1+\epsilon \xi^{\ell,-}_{ij})w^{\ell,-}_{ij} \right)z + (1+\epsilon \xi^{\ell,-}_i) b^{\ell,+}_i-(1+\epsilon \xi^{\ell,-}_i) b^{\ell,-}_i,
\end{align*}
and $\xi^{\ell,\pm}_{ij}$, $\xi^{\ell,\pm}_{i}$ and $\xi^\ell_i$ are independent standard normals defined in \eqref{eq:normal} 
(here we suitably re-enumerated $\eta_j$ in \eqref{eq:normal}). As shown in Section \ref{sec:one pot}, due to the difference of the time scale between the dynamics of $x^\ell_i$'s and the parameters $w^\ell_{ij}$'s and $b^\ell_i$'s, we  assume that $w^\ell_{ij}$'s and $b^\ell_i$'s are fixed for $t\in [0,T]$. 
Using the feedback-free structure considered in Remark \ref{rem:feedback free}, we can inductively prove the following proposition for the variation of the species $x^\ell_i(t;\epsilon)$'s from the noise. See Appendix \ref{app:proof} for the proof.
\begin{prop}\label{prop:deviation of forward}
    Under the above setting, there exists a finite random variable $\Theta$ such that for each $\ell$ and for $i$, $\sup_{t\in [0,T]}|x^\ell_i(t;\epsilon)-x^\ell_i(t)| \le \epsilon \Theta$,
where 
\begin{align*}
    \Theta=\Theta(L, \max_{\ell}M^\ell,h,\max_{\ell,i,j}w^{\ell,\pm}_{ij},\max_{\ell,i}b^{\ell,\pm}_i,\sup_{\ell,i,j}\sup_{t\in ]0,T]}(|\xi^{\ell,\pm}_{ij}|+|\xi^{\ell,\pm}_{i}|+|\xi^\ell_{i}|)).
\end{align*}
\end{prop}
Using the result of Proposition \ref{prop:deviation of forward}, we can similarly obtain the deviation of the backward CRN. Let $LW^{\ell,\pm}_{ij}$ and let $LB^{\ell,\pm}_{i}$ be the species  corresponding to $\frac{\partial \mathcal L}{\partial w^{\ell}_{ij}}$ and $\frac{\partial \mathcal L}{\partial b^{\ell}_{i}}$ in the dual rail manner, respectively. We denote their concentrations by $lw^{\ell,\pm}_{ij}(t;\epsilon)$ and $lb^{\ell,\pm}_i(t;\epsilon)$ for $(\Sp,\C,\Re,\mathcal K^\epsilon)$, and we also denote their concentrations by $lw^{\ell,\pm}_{ij}(t)$ and $lb^{\ell,\pm}_i(t)$ for $(\Sp,\C,\Re,\mathcal K)$.
\begin{prop}\label{prop:devevation of backward}
    Under the same setting for Proposition \ref{prop:deviation of forward}, there exists a finite random variable $\bar \Theta$ such that for each combination of indices $\ell,i,j$,
\begin{align*}
   \sup_{t\in [0,T]}(|lw^\ell_{ij}(t;\epsilon)-lw^\ell_{ij}(t)|+|lb^\ell_i(t;\epsilon)-lb^\ell_i(t)|) \le \epsilon \bar \Theta,
\end{align*}  
where  $\bar \Theta=\bar \Theta(\Theta, L, \max_{\ell}M^\ell,h,\max_{\ell,i,j}w^{\ell,\pm}_{ij},\max_{\ell,i}b^{\ell,\pm}_i)$.
\end{prop}




\begin{figure}
    \centering
\includegraphics[width=1\linewidth]{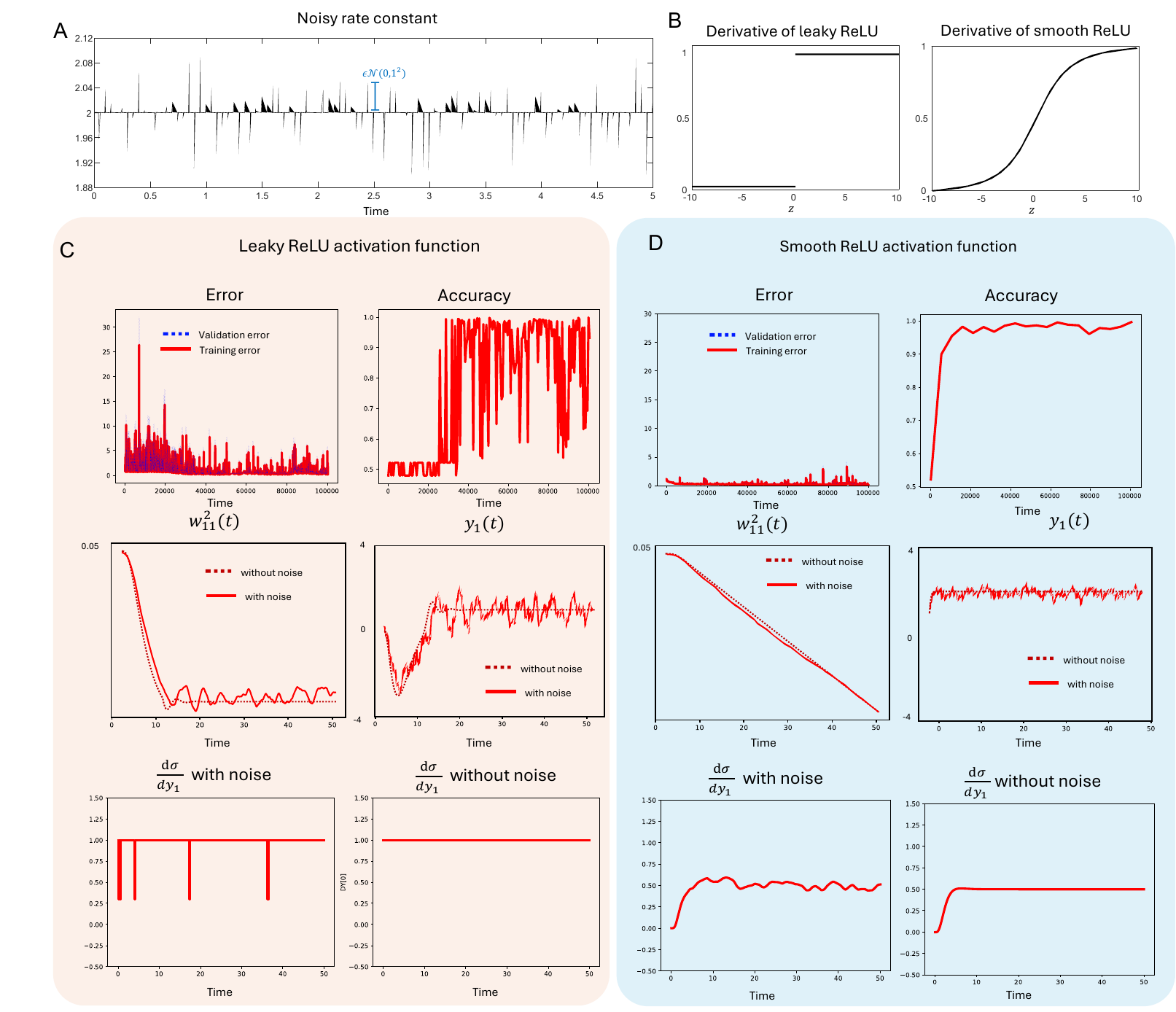}
    \caption{A. Time evolution of the noisy rate constant. B. The derivatives of the activation functions. C. The time evolution of the training and the validation error functions and the accuracy functions obtained by a CRN for a NN with the leaky ReLU (the first and second panels) and a CRN for a NN with the smooth ReLU (the third and fourth panels). D. The time evolutions of the weight and the output obtained by a CRN for a NN with the leaky ReLU (the first and second panels) and a CRN for a NN with the smooth ReLU (the third and fourth panels). E. The derivatives of the activation functions with noise and without noise. The first two panels are derivatives of the leaky ReLU and the others are the derivatives of the smooth ReLU. See Table \ref{table:parameters} in Supplementary material for the parameters used for these simulations.}
\label{fig:activation_noise}
\end{figure}

\subsection{Comparison to CRNs implementing NNs with non-smooth activation functions}


In this section, we demonstrated that usage of smoothed activation functions can make CRNs implementing NNs more robust to noise compared to CRNs implementing NNs with non-smooth activation functions such as the leaky ReLU. Specifically, the gradient of a non-smooth activation function has discrete jumps, which can lead to significant errors.

We built two CRNs implementing NNs whose activation functions are the leaky ReLU and the smooth ReLU with $h=4$, respectively. To build a NN equipped with the leaky ReLU, we also used the CRNs for the arithematic operations proposed in \cite{buisman2009computing}. In contrast to the approach in \cite{lakin2023design}, where clock proteins were used to achieve the leaky ReLU function, we opted for using built-in characteristic functions for simplicity. This means that the leaky ReLU is computed without involving chemical species in this paper.

For this comparison, we used MNIST data sets.  
The rate constants in backward CRNs fluctuate with the independent noise described in \eqref{eq:normal} (Figure \ref{fig:activation_noise} A). We chose $t_0=0.05$ and $t_1=0.005$. To realize intrinsic noise, we also added the same independent noise to the nodes in the hidden and the output layers on each interval $[kt_0,kt_0+t_1)$ for $k=0,1,\dots,$.

Figure \ref{fig:activation_noise} C and D displayed the fluctuation of the error and the accuracy functions. Due to noise, a long simulation may end up blow-up. Hence, to avoid blow-up, we ran both CRNs until time $50 \times 2000$, which is equivalent to $2000$ iterations with randomly changed inputs. We set the noise size $\epsilon=0.01$. The training error, validation error and accuracy displayed higher fluctuations in the NN with the leaky ReLU.  Furthermore, we compare the time evolution of the species $w^2_{11}(t)$ and $y_1(t)$ for the two CRNs under the noisy rate constants to those without noise (Figure \ref{fig:activation_noise} C and D middle). We use a single input and run both CRNs for $[0,50)$, which is an enough time for convergence for both CRNs. The time trajectories of those species in the CRN implementing the leaky ReLU fluctuate more than the trajectories of the CRN implementing the smooth ReLU. For these plots, we used a noise size $\epsilon=0.1$ to clearly visualize the differences in dynamics.

In Figure \ref{figs1}, we provided additional comparisons of the time evolutions of the error to support the noise-robustness of our CRN.


\subsection{Sensitivity to noise level and the running time}
We here examined how the performance of the CRNs implementing a NN changes with different noise levels and the running time for each iteration. We found that a NN with the leaky ReLU is more sensitive to the noise level and the running time than a NN with the smooth ReLU. For both CRNs the train error function oscillates due to noise (Figure \ref{fig:activation_noise}), and hence to clear comparison we use an average train error defined as 
\begin{align*}
    \frac{1}{10}\sum_{i=0}^9 E((K-i)t),
\end{align*}
rather than the train error at a single time point $KT$. Recall that $T$ is the running time of a single iteration and the $K$ is the total number of iterations. 

The same noise as in \eqref{eq:normal} was added to each the rate constants of the backward CRN $(\Sp^b,\C^b,\Re^b,\K^b)$ and species $X^\ell_i$'s of the forward CRN $(\Sp^f,\C^f,\Re^f,\K^f)$ for both the CRNs implementing NNs with the leaky ReLU and the smooth ReLU. We chose $\epsilon\in\{0, 0.02,0.04,\dots, 0.1, 0.12\}$. 
 The MNIST data set was employed for this test.
 For each $\epsilon$, the running time for each iteration is $50$. We measured the training error and the validation error by four independent experiments. To avoid potential blow-up especially for the case of a high noise level, we iterated the CRNs only $200$ times. 
The running time $T$ for each iteration is $50$. Figure \ref{fig:sensitivity} shows that the NN with the smooth ReLU is much more robust to different noise levels that the NN with the leaky ReLU.

We used non-linear regression data fitting to test sensitivity to noise levels. We used training data points $\{(x_i,\sin(x_i))\}_{i=1}^{13}$ for $x_i\in [0,\pi)$ to fit the sine curve using CRNs implementing NNs with leaky ReLU and smooth ReLU, respectively. The four independent training errors demonstrate the robustness of the CRN implementing the smooth ReLU (Figure \ref{fig:sensitivity}). We used 1300 iterations, with each iteration taking 50 units of running time. Due to the small number of data points, we did not separately use validation data points.

Next, we tested the robustness to the choice of the running time for each iteration. In the absence of noise, it is obvious that the accuracy increases as the running time increases, since the species of the CRNs have enough time to converge to their steady state. However, when noise exists, the accuracy may suffer from a long running time as the effect of noise can accumulate over time. Since users, especially wet lab experimentalists, might not know the optimal running time in advance, it is crucial to develop a CRN that is robust to running time variations. By varying the running times, we measured the training error and the validation error of the CRNs implementing NNs with the leaky ReLU and the smooth ReLU for MINST data sets. Figure \ref{fig:sensitivity} C shows that the CRN for the smooth ReLU is much more robust to the running time than the CRN for the leaky ReLU. We used $200$ iterations. As previously, we also used the sine curve regression to compare two CRNs, and the error turned out less sensitive for the case of the smooth ReLU (Figure \ref{fig:sensitivity} D). In this experiment, $1300$ iterations was used. The training error for the CRN calculating the smooth ReLU is more robust than that of the CRN calculating the leaky ReLU, while the training error is relatively high with the running time of $10$.

\begin{figure}[!h]
    \centering
    \includegraphics[width=1\linewidth]{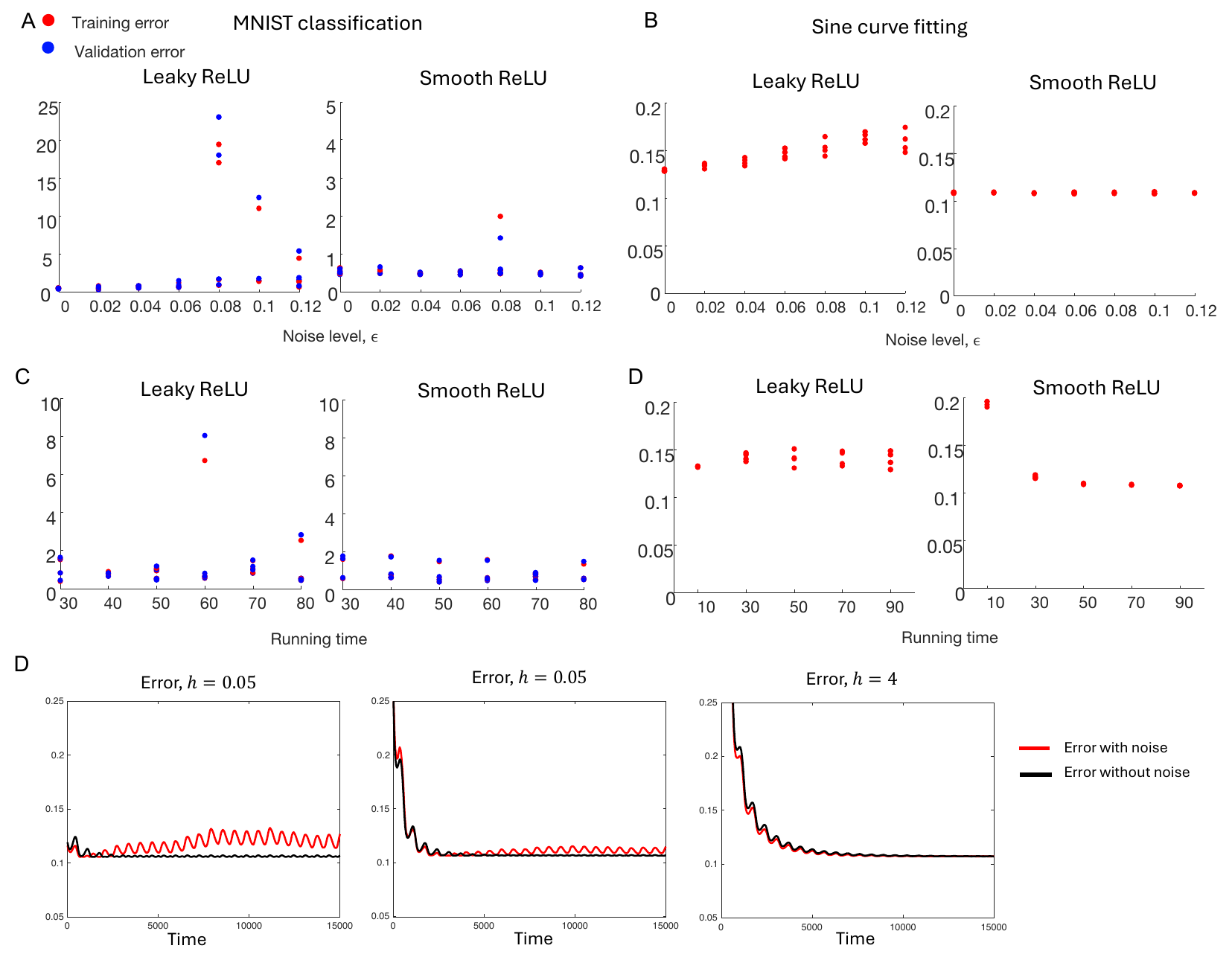}
    \caption{Sensitivity to the noise level (A and B) and the running time (C and D) of the Leaky ReLU and Smooth ReLU measured with the MNIST data set and the sine curve fitting. D. The time evolution of the error with the different values of $h$. See Table \ref{table:parameters} in Supplementary material for the parameters used for these simulations.}
    \label{fig:sensitivity}
\end{figure}

\subsection{Tuning the smoothing parameter to enhance robustness}
The value of $h$ in the smooth ReLU function affects the change in slope around the origin. When $h$ is large, the activation function and its derivative are less affected by noise. We tested the accuracy of the proposed CRN for implementing a neural network to classify MNIST datasets (see Figure \ref{fig:sensitivity} D). The training error fluctuates less for higher values of $h$. However, this comes at the cost of slower convergence, as the function's less dramatic changes near the origin can limit its activation power. Due to its reduced activation for larger h values, the smooth ReLU may struggle to adapt to significant input changes, even without noise. As demonstrated in Figure S, the smooth ReLU with $h=1$ outperformed $h=4$ on the MNIST data. Therefore, when dealing with potentially dramatic input variations, carefully adjusting $h$ can be needed for a better performance.

 \begin{rem}
    While adjusting the leaky ReLU's slopes can theoretically enhance smoothness, this approach is overly drastic. Altering $\alpha$ and $\beta$ in \eqref{eq:leaky Relu} uniformly impacts the behavior of the function on all large $|x|$, potentially hindering activation and performance (Figure \ref{figs3}). In contrast, the smooth asymptotic slopes of ReLU remain consistent for large $|x|$, ensuring a more balanced and effective activation (i.e. $\lim_{x\to\infty} \frac{\sigma(x)}{x}=1$ and $\lim_{x\to-\infty} \frac{\sigma(x)}{x}=0$).
\end{rem}

\section{Discussion}
We have seen that chemical reactions can be used to mimic artificial NNs, but we need better methods for dealing with the external and internal noise of these reactions when using them for neural network computing and training. In our research, we proposed that creating smooth activation functions using these chemical reactions could improve the reliability of neural network computations. 

  We have developed bimolecular CRNs that can compute the gradients of loss functions when the activation functions are smooth functions of the nodes. Instead of discrete iterations, the parameters are continuously updated, allowing for one-pot computing with the entire CRN. We have proved the exponential convergence of the proposed CRN.
Simulations have demonstrated that our CRN is more robust to noise in the rate parameters and inputs compared to another CRN that implements neural networks with non-smooth activation functions.


Interesting ideas about rate-independent chemical reaction networks for a feed-forward NN were introduced suggesting that fine-tuning of the reaction rates is not necessary \cite{senum2011rate, vasic2022programming}. This means that the fluctuation of the reaction rates does not alter the chemical computations for feed-forward NNs leading to a ideal solution for noise-robust biochemical computations. However, this type of biochemical computing can only realize semilinear functions \cite{chen2014rate, chen2023rate, chen2014deterministic}. 
Hence noise-robustness is unclear for the rate-independent chemical reaction networks if the backpropagation, which often requires non-linear functions, is also implemented by chemical reactions. 
Nonetheless, there is potential for improvement in robustness by partially employing such rate-independent computation. The addition, subtraction, minimum, and maximum functions are realized by rate-independent steady states of some CRNs \cite{chen2014rate}. We can replace some sub-CRNs of our CRN, whose steady states are rate-dependent, by the CRNs suggested by \cite{chen2014rate} to enhance robustness to noise within the rate parameters. 

The convergence rate is another aspect that needs improvement. Currently, while the convergence of the species in the forward and backward CRNs is exponentially fast in time, the convergence speed can be delayed by many system components such as the reaction rates and the initial amounts of the species. In particular, the initial amount of the species can be controlled, but it is uncertain which initial conditions can minimize the convergence time. A recently suggested study \cite{anderson2024chemical} proposed new CRN structures for initial amount-independent convergence. This idea has the potential to improve the robustness of the convergence speed of the CRNs used for computing and training neural networks.


The future aim will be to move beyond computer simulations and actually create the proposed chemical reaction network in a laboratory setting and within living organisms. To use CRNs to carry out deep neural networks effectively, we may need to simplify the structure of the chemical reaction network. Currently, the number of reactions and species required to calculate the gradient of the smooth activation functions is not small. One potential solution is to use the matrix multiplication with DNA suggested in \cite{oliver1997matrix}. In our CRN, the matrix multiplication for the weighted sum $wx+b$ was computed directly using reactions for the arithmetic operations. However a more refined approach for matrix multiplications would be used to reduce the number of reactions.

\appendix

\section{Speed of convergence}\label{app:conv speed}

In this section, we show that the dynamical system associated with the CRN implementing the backward network admits exponential convergence to its steady state provided that the concentration of the species $W^{\ell,\pm}_{ij}$ and $B^{\ell,\pm}_i$ in the update network are fixed at constants. 

The forward and backward CRNs are made up of sub-CRNs, each of which is responsible for performing an arithmetic operation. Importantly, there are no feedback loops within the network (See Figure \ref{fig: structure} and Remark \ref{rem:feedback free}). Therefore, the convergence of the species in the backward CRN does not impact the convergence of the species in the forward CRN. This is crucial for proving  the convergence of the backward CRN, as previously shown in the convergence of the species in the forward CRN  \cite{anderson2021reaction}.

  We note that each species in the backward CRN $(\Sp^b,\C^b,\Re^b,\K^b)$, say $R\in \Sp^b$, is governed by a  differential equation with some smooth functions $\phi:\mathbb R^{m}_+\to \mathbb R$ and  $\psi:\mathbb R_+\to \mathbb R$ such that
  \begin{align}\label{eq:r}
      \frac{d}{dt}r(t)=\phi(s_1(t),\dots,s_{m-1}(t))-\psi(s_m(t))r(t),
  \end{align}
  where $m$ is some positive index, and $s_i$'s are the concentration of species whose exponential convergence to a positive steady state is already guaranteed. Furthermore it holds that 
  \begin{align*}
      \lim_{t\to \infty}\phi(s_1(t),\dots,s_{m-1}(t))>0 \quad \lim_{t\to \infty}\psi(s_m(t))>0.
  \end{align*}
For example, from the reactions in Table \ref{table:1}, we can see that the concentration $dy(t)$ of $DY$ at time $t$ follows
\begin{align}\label{eq:dy}
    \frac{d}{dt}dy(t)=(y(t))^2-a^2(t)dy(t),
\end{align}
  where $y(t)$ is the concentration of the species corresponding to the node in the output layer of the forward network, and $a^2(t)$ is the concentration of the auxiliary species $A^2$ in Table \ref{table:1}. Both are feeding forward $DY$ as catalysts. The exponential convergence of $y(t)$ is shown in \cite{anderson2021reaction}, and the exponential convergence of $a^2(t)$ ($2$ is a super index, not a power) can be  shown by its differential equation
  \begin{align*}
      \frac{d}{dt}a^2(t)=(y(t))^2+h-a^2(t), 
  \end{align*}
  whose solution can be written as
$a^2(t)=e^{-t}\int_0^t e^s \left ( (y(s))^2+h\right ) ds+a^2(0)e^{-t}$. Similarly
using the exponential convergence of $y(t)$ and $a^2(t)$, we can obtain exponential convergence to a positive value for $dy(t)$ from \eqref{eq:dy}. We generalize this in the following lemma.

\begin{lem}\label{lem:expo}
    Let $r:[0,\infty)\to \mathbb R_+$ be a solution to \eqref{eq:r}. If for each $i=1,\dots,m$, there exist positive constants $s_{i,\infty}$ and $\beta>0$ such that $|s_i(t)-s_{i,\infty}|\le e^{-\beta t}$ for all $t$, where $\lim_{t \to \infty}s_i(t):=s_{i,\infty}$. Then there exist $\alpha>0$, $c>0$ and $r_\infty>0$ such that $|r(t)-r_\infty|\le ce^{-\alpha t}$ for all $t$.
\end{lem}
\begin{proof}
For simplicity we denote $\phi(t)=\phi(s_1(t),\dots,s_{m-1}(t))$ and $\psi(t)=\psi(s_m(t))$. 
    We also denote $\phi_\infty=\dlim_{t\to \infty}\phi(s_1(t),\dots,s_{m-1}(t))$ and $\psi_\infty=\dlim_{t\to \infty}\psi(s_m(t))$. These limits exist as $\phi$ and $\psi$ are continuous functions. Letting $V(t)=\frac{1}{2}(\phi_\infty -\psi_\infty r(t))^2$, we have
    \begin{align*}
        \frac{d}{dt}V(t)&=(\phi_\infty -\psi_\infty r(t))(-\phi(s_1(t),\dots,s_{m-1}(t))+\psi(s_m(t))r(t))\\
     &= (\phi_\infty -\psi_\infty r(t))\left (-(\phi(t)-\phi_\infty)+(\psi(t)-\psi_\infty) r(t)-\phi_\infty +\psi_\infty r(t) \right )  \\
     &=-(\phi_\infty -\psi_\infty r(t))^2+
     \left (-(\phi(t)-\phi_\infty)+(\psi(t)-\psi_\infty) r(t) \right ).
    \end{align*}
    Note that due to the boundedness of $r(t)$ given by \eqref{eq:r} and the smoothness of $\phi$ and $\phi$, we have
    \begin{align*}
C|(\phi(t)-\phi_\infty)+(\psi(t)-\psi_\infty) r(t)|\le \sup_{i}|s_i(t)-s_{i,\infty}|+|s_m(t)-s_{m,\infty}|
    \end{align*}
    for some $C$. Here we applied the Mean value theorem to $\phi$ and $\psi$. Therefore 
    there exists $\gamma>0
$ such that $(\psi(t)-\psi_\infty) r(t)\le Ce^{-\gamma t}$ for all $t$ by the exponential convergence of $s_i$'s. Thus $V'(t)\le -2V(t)^2+e^{-\gamma t}$. By a version of the Gronwall's inequality \cite{anderson2021reaction}, we have
\begin{align*}
    V(t)\le V(0)e^{-2t}+e^{-2t}\int_0^t e^{2s}e^{-\gamma s}ds,
\end{align*}
and note that the right-hand side converges to $0$ exponentially fast. Therefore we can find $c>0$ and $\alpha>0$ such that $|r(t)-r_\infty|\le ce^{-\alpha t}$, where $r_\infty=\frac{\phi_\infty}{\psi_\infty}$. 
\end{proof}

Hence by Lemma \ref{lem:expo}, it is sufficient to show that the exponential convergence to positive steady states for the catalyst species corresponding to $s_i$'s in \eqref{eq:r}. The catalyst species $s_i$'s in the differential equations \eqref{eq:r} for the dynamic species $LY, DY, LW^{\ell-1}_{ij}$ and $LZ^{\ell-1}_i$ in $(\Sp^b,\C^b,\Re^b,\K^b)$ are given by either from the forward networks or the auxiliary species such as $A^\ell_i$'s in Table \ref{table:1}. The species from the forward network converge exponentially fast \cite{anderson2021reaction}, and as shown above so the auxiliary species do. Recursively, the $s_i$'s in \eqref{eq:r} when $r$ is either $DX^{\ell-1}_i$ or $LX^{\ell-1,\pm}_i$ are also the nodes $X^{\ell-1}$ and $A^{\ell-1}_i$, which converge exponentially as mentioned above. This leads that exponential convergence of $DX^{\ell-1}_i$'s and $LX^{\ell-1,\pm}_i$'s by Lemma \ref{lem:expo}.  Then we further recursively show that $LW^{\ell-2,\pm}_{ij}$'s and $LZ^{\ell-2,\pm}_{i}$ converge exponentially using Lemma \ref{lem:expo} since the catalyst species are either the species from the forward network or the species $DX^{\ell-1}_i$'s and $LX^{\ell-1,\pm}_i$'s, whose exponential convergence is already guaranteed. Repeating this, we can show exponential convergence of all the species in $(\Sp^b,\C^b,\Re^n,\K^b)$. We highlight again that the feedback-free structure is the key to such a recursive argument. \hfill $\square$


Anderson and Joshi showed more delicate construction to control the convergence rates that can be independent to the initial condition \cite{anderson2024chemical}. We did not employ the construction as the proposed CRNs are not necessarily bimolecular.

\section{Proof of Propositions  in Section \ref{sec:noise analysis}}\label{app:proof}
Throughout this section, the initial conditions for each species in $(\Sp,\C,\Re,\mathcal K^\epsilon)$ and $(\Sp,\C,\Re,\mathcal K^\epsilon)$ are identical. We also assume $\epsilon \in (0,1)$ sufficiently small.

\begin{lem}
Under the same setting as Proposition \ref{prop:deviation of forward}, we have
\begin{align}
    &\sup_{\ell,i}\sup_{t\in [0,T]}(x^\ell_i(t;\epsilon)+x^\ell(t)):=R<\infty \quad \text{and} \label{eq:upper bound of x}\\
     &\inf_{\ell,i}\inf_{t\in [0,T]}(x^\ell_i(t;\epsilon)+x^\ell(t)):=r>0 \label{eq:lower bound of x}
\end{align}
\end{lem}
\begin{proof}
Since the inputs $\hat x_i$'s and the parameters $w^{\ell,\pm}_{ij}$ and $b^{\ell,\pm}_i$ are fixed, from \eqref{eq:forward eq}
\begin{align*}
    \frac{d}{dt}x^1_i(t;\epsilon)\le c_0-2(x^1_i(t;\epsilon))^2\le c_1-c_2 x^1_i(t;\epsilon)
\end{align*}
for some positive constants $c_i$'s. Thus by Gronwall's inequality $\sup_{t\in [0,T]}\sup_ix^1_i(t;\epsilon)<\infty$. Then we can inductively have  $\sup_{t\in [0,T]}\sup_ix^2_i(t;\epsilon)<\infty$. When $\epsilon=0$, we can get the same result, and hence eventually we have \eqref{eq:upper bound of x}. Furthermore, due to the term $h$
in \eqref{eq:forward eq} and \eqref{eq:forward eq2}, both $x^\ell_i(t;\epsilon)$ and $x^\ell_i(t)$ do not approach to zero. Hence \eqref{eq:lower bound of x} follows. 
Notably $R=R(\max_{\ell,i,j}w^{\ell,\pm}_{ij},\max_{\ell,i}b^{\ell,\pm}_i)$.    
\end{proof}

\textbf{Proof of Proposition \ref{prop:deviation of forward}}
To prove the result inductively, assume that for a fixed $\ell$, there exists $\Theta^{\ell-1}$ such that for any $i$, we have that $\sup_{t\in [0,T]}|x^{\ell-1}_i(t;\epsilon)-x^{\ell-1}_i(t)|\le \epsilon \Theta^{\ell-1}$, where
\begin{align}\label{eq:theta dependence}
    \Theta^{\ell-1}&=\Theta^{\ell-1}(L, \max_{\ell}M^m,h,\max_{\ell,i,j}w^{m,\pm}_{ij},\max_{\ell,i}b^{m,\pm}_i,\sup_{\ell-1,i,j}\sup_{t\in ]0,T]}(|\xi^{m,\pm}_{ij}|+|\xi^{m,\pm}_{i}|+|\xi^m_{i}|)).
\end{align}
For small enough $\epsilon$, by \eqref{eq:forward eq} and \eqref{eq:forward eq2} we have
\begin{align}
    &\frac{d}{dt}(x^{\ell}_i(t;\epsilon)-x^{\ell}_i(t))\\
    &\le \epsilon\left(\sum_j(w^{\ell-1,+}_{ij}-w^{\ell-1,-}_{ij})(x^{\ell-1}_i(t;\epsilon)-x^{\ell-1}_i(t)) +C_1^{\ell-1}  \right )-((x^{\ell}_i(t;\epsilon))^2-(x^{\ell}_i(t))^2) \notag
    \\
    &\le \epsilon \left (C_2^{\ell-1} \Theta^{\ell-1}+C^{\ell-1}_2\right )-(x^{\ell}_i(t;\epsilon)+x^{\ell}_i(t))(x^{\ell}_i(t;\epsilon)-x^{\ell}_i(t)), \notag
\end{align}
where the constants $C^{\ell-1}_1$ and $C^{\ell-1}_2$ have the same dependence as \eqref{eq:theta dependence}.
Then 
\begin{align}
    \frac{d}{dt}(x^{\ell}(t;\epsilon)-x^{\ell}(t))^2&=(x^{\ell}(t;\epsilon)-x^{\ell}(t))\frac{d}{dt}(x^{\ell}(t;\epsilon)-x^{\ell}(t))\notag \\
    &\le \epsilon R \left (C_2^{\ell-1} \Theta^{\ell-1}+C^{\ell-1}_2\right ) - (x^{\ell}_i(t;\epsilon)+x^{\ell}_i(t))(x^{\ell}_i(t;\epsilon)-x^{\ell}_i(t))^2 \notag \\
    &\le \epsilon R \left (C_2^{\ell-1} \Theta^{\ell-1}+C^{\ell-1}_2\right ) - r(x^{\ell}_i(t;\epsilon)-x^{\ell}_i(t))^2. \label{eq:finish the deviation of x}
\end{align}
Thus by Gronwall's inequality, we get $|x^{\ell}_i(t;\epsilon)-x^{\ell}_i(t)|\le \epsilon \Theta^{\ell}$ where $\Theta^{\ell}$ is a positive constant satiying the same dependent as \eqref{eq:theta dependence}. Therefore, it suffices to show the hypothesis $|x^1_i(t;\epsilon)-x^1_i(t)|\le \epsilon \Theta^1$ to finish the proof inductively. When $\ell=1$, by \eqref{eq:forward eq} and \eqref{eq:forward eq2}, we have
\begin{align*}
    \frac{d}{dt}(x^{1}(t;\epsilon)-x^{\ell}(t))&\le C^0 - (x^{1}_i(t;\epsilon)+x^{1}_i(t))(x^{1}_i(t;\epsilon)-x^{1}_i(t))\\
\end{align*}
where $C^0$ is a positive constant solely depending on the inputs $w^{0,\pm}_{ij}, b^{0,\pm}_i$ and the noise terms $\xi^{0,\pm}_{ij}$, $\xi^{0,\pm}_{i}$ and $\xi^{0}$. As these are bounded (with probability 1), we can deduce 
$|x^1_i(t;\epsilon)-x^1_i(t)|\le \epsilon \Theta^1$ using the same argument as \eqref{eq:finish the deviation of x} with Gronwall's inequality. \hfill $\square$

The proof of Proposition \ref{prop:devevation of backward} is basically the same as the proof of Proposition \ref{prop:deviation of forward}. The gradient $\nabla_{w^\ell}\mathcal L$ and $\nabla_{b^\ell}\mathcal L$ are recursively obtained as \eqref{eq:derivative of L wrt x}--\eqref{eq:derivative of L wrt b}. Due to the recursive relations, the fluctuation of the gradients can be quantified with the fluctuation of the terms $\frac{\partial x^{\ell+1}_i}{\partial z^{\ell}_i}$ when $w^{\ell,\pm}_{ij}$ and $b^{\ell,\pm}_{i}$'s are fixed. Furthermore, the fluctuation of the derivatives is completely computed by the fluctuation of $x^\ell_i$ by \eqref{eq:derivative of relu}, which is shown in Proposition \ref{prop:deviation of forward}. By smoothness of $\sigma(z)$ with respect to $z$, the fluctuation of the $\frac{\partial x^{\ell+1}_i}{\partial z^{\ell}_i}$ is obtained, and the eventually desired result in Proposition \ref{prop:devevation of backward} follows. In this vein,  $\bar \Theta$ in Proposition \ref{prop:devevation of backward} depend on $\Theta$,  the number of layers and the fixed parameters of the weights and the biases.

\section{Implement neural networks with general activation functions}\label{app:other act}
Our implementation of the backpropagation algorithm is not limited to the smoothed ReLU function. It extends to accommodate various activation functions. 
In the context of a generic NN, the backpropagation process entails gradient calculation through arithmetic operations on the outputs. Thus, as long as the forward network is implemented as a CRN, the entire learning process of the neural network can be realized by constructing a backward network that employs the appropriate arithmetic operations based on the species containing value of output of neural network as suggested in this paper.

For instance, consider a NN employing the sigmoid function $\sigma(z)=\frac{1}{1+e^{-z}}$ as its activation function, implemented via a CRN. In this scenario, differentiation of the activation function yields that
\begin{align*}
    \sigma'(z)=\frac{e^{-z}}{(1+e^{-z})^2}=\frac{1+e^{-z}-1}{(1+e^{-z})^2}=\sigma(z)-\sigma_{\text{sig}}(z)^2.
\end{align*}
Then if $x^{\ell+1}_j=\sigma(z^\ell_j)$ for each $\ell$ and $j$, then
\begin{align*}
    \nabla_{z^\ell}x^{\ell+1}=\left (x^{\ell+1}_1-(x^{\ell+1}_1)^2,\dots,x^{\ell+1}_{C^\ell}-(x^{\ell+1}_{C^\ell})^2 \right )
\end{align*}
Thus the gradients are composed of arithmetic operations on the outputs, allowing us to construct another backward CRN using our method, utilizing the chemical species corresponding to the outputs.

In \cite{arredondo2022supervised} a CRN was built to calculate the hyperbolic tangent activation functions and to train a NN with the activation function.
Note that it is challenging to directly compute exponential functions that are needed to compute $\tanh(z)=\frac{e^{z}-e^{-z}}{e^{z}+e^{-z}}$. However, the authors in \cite{arredondo2022supervised} used the relation $\tanh'(z)=1-\tanh^2(z)$ to compute $\tanh(z)$ with clock proteins that are used to run sub-CRNs separately. Then for training, we can build a backward CRN similarly to calculate \begin{align*}
    \nabla_{z^\ell}x^{\ell+1}=\left (1-(x^{\ell+1}_1)^2,\dots,1-(x^{\ell+1}_{C^\ell})^2\right ).
\end{align*}


\section{Details of the data and simulations}\label{app:data}
Here we give more details on the data sets we used in Section \ref{sec:Xor, Iris}.

\subsection{XOR}
Our objective is to construct an XOR classifier using a regression approach applied to the data set, which comprises tuples of the form $\{(\hat x,\hat y) : ((0,0),0),((1,0),1),((0,1),1),((1,1),0)\}$. Here $\hat x\in \{0,1\}^2$ is the input and $\hat y\in \{0,1\}$ is the output. The neural network takes $\hat x$ as input and produces $y$ as the output. Subsequent to the training process, we determine whether the model's output surpasses the threshold of 0.5 to classify the XOR operation correctly.
To achieve this objective, we used the CRN that implements a two-layered NN with one hidden layer. The initial layer comprises two nodes, the subsequent hidden layer consists of 10 nodes, and the final output layer has a single node. That is, $M=3, C^1=2, C^2=10$ and $C^3=1$ in terms of the notations of Section \ref{sec:nn}.

\subsection{Iris}\label{app:iris}
The Iris dataset stands as one of the most renowned datasets in the realm of machine learning \cite{Iris}. This dataset comprises four distinctive features related to iris flowers belonging to three distinct species. In the context of our neural network architecture, the input is represented as a single vector in $\mathbb{R}^4$, and the output is a single vector in $\mathbb{R}^3$. We employ one-hot encoding to facilitate the comparison between our network's output and the ground truth data.
Our neural network architecture has the initial layer consisting of four nodes, followed by a single hidden layer comprising 10 nodes, and the output layer consisting of three nodes. 135 data points are used for training and a total 67500 iterations are used. 

We provided pair plots in Figure \ref{fig:application to data} B bottom. The off-diagonal plots provide the classification projected onto the plane defined by feature 1 on the $x$-axis and feature 2 on the $y$-axis. The diagonal plots show the distribution of the accumulated classification data over the same column, where the $x$-axis is feature 1 and the $y$-axis is the probability density (mass). See more details about the pair plots in \cite{emerson2013generalized}.

\subsection{MNIST}
The MNIST (Modified National Institute of Standards and Technology) dataset is one of the most well-known machine learning datasets, consisting of handwritten digit images \cite{MNIST}. 
For our experimental purposes, we utilized a downscaled version of subset of the MNIST dataset, comprising exclusively of images representing the digits 0 and 1. These images were resized to an 8x8 pixel square format. 
Our neural network architecture is designed with a configuration similar to that discussed in Section \ref{sec: backward network} Specifically, the initial layer consists of 64 nodes, followed by a hidden layer comprising 4 nodes, an an output layer consisting of two nodes. 1000 data points are used for training and 2000 iterations were used. 

\subsection{Accuracy function}
The XOR, Iris, and MNIST data sets are for classification. For each output $\hat y$, let $f(\hat y)$ be the classification. 
For the output of the XOR data point $\hat y$,
\begin{align*}
    f(\hat y) =
    \begin{cases}
     0 \quad &\text{if $\hat y \in \{(0,0),(1,1)\}$,}\\
     1 &\text{otherwise.}
    \end{cases}
\end{align*}
With given input and output data points $\{(\hat x_n, \hat y_n)\}_{n=1}^N$, the accuracy is given by \begin{align*}
    A(t)=\frac{1}{N}\sum_{n=1}^N \mathbbm 1_{|y_n(t)-f(\hat y_n)| \le 1/2},
\end{align*}
where $y_n(t)$ is the output of the forward network with the $n$ th input $\hat x_n$. Where the output of the NN is calculated with the weight and bias $w^\ell_{ij}(t)$ and $b^\ell_{i}(t)$, which are the time evolutions given by the backward and the update networks with the training data. 
As $\mathbbm 1_{y_n(t)-f(\hat y_n)}$ measures the accuracy of the classification of the CRN-implemented NN, $A(t)$ is the average accuracy over $N$ data points.

For both the Iris and the MNIST cases, the accuracy was evaluated by one-hot encoding. That is, we compare the argmax of $y$ and $\hat y$ whether it coincides or not. For example, for an input $\hat x$ that represents versicolor, $\hat y$ is $(1,0,0)$. If our neural network output gives $y=(0.9,0.05,0.1)$ from the input $\hat x$, then we will have $\mathbbm 1_{\{argmax\{y(t)\}=argmax\{\hat y\}\}} = \mathbbm 1_{\{1=1\}}=1$. Thus the accuracy is defined as
\begin{align*}
    A(t)=\frac{1}{N}\sum_{n=1}^N \mathbbm 1_{\{argmax\{y_n(t)\}=argmax\{\hat y_n\}\}},
\end{align*}
where $N$ is the number of the total input data set.

The evaluation of accuracy was performed at each $100$ th iteration.

\begin{table}[!h]
\centering
{\fontsize{5}{10}\selectfont  
    \resizebox{\textwidth}{!}{
\begin{tabular}{|c|c|c|c|c|c|c|c|}
\hline
Parameters  & Figure 6A & Figure 6B & Figure 6C & Figure 7CD & Figure 8AC & Figure 8 BC & Figure 8D\\
\hline
\makecell{Learning rate \\ $\eta \times T$} & $0.002\times 300$ & $0.00025\times 300$ & $0.0001\times 500$ & $0.0001\times 50$ & $0.0001\times 50$ & $0.0001\times 50$ &  $0.001\times 50$\\
\hline
\makecell{Number of \\ iterations} & $36000$ & $54000$ & $27560$ & $2000$ & $200$ & $1300$ & $650$\\
\hline
\makecell{Number of\\ training data sets} & $4$ & $135$ & $13780$ & $1000$ & $100$ & $13$ &  $13$ \\
\hline
\makecell{Number of \\ validation data sets} & $4$ & $15$ & $1000$ & $1000$ & $100$ & None & None \\
\hline
\end{tabular}}}
\caption{Parameters used for simulation.}\label{table:parameters}
\end{table}

\begin{figure}[!h]
    \centering
    \includegraphics[width=1\linewidth]{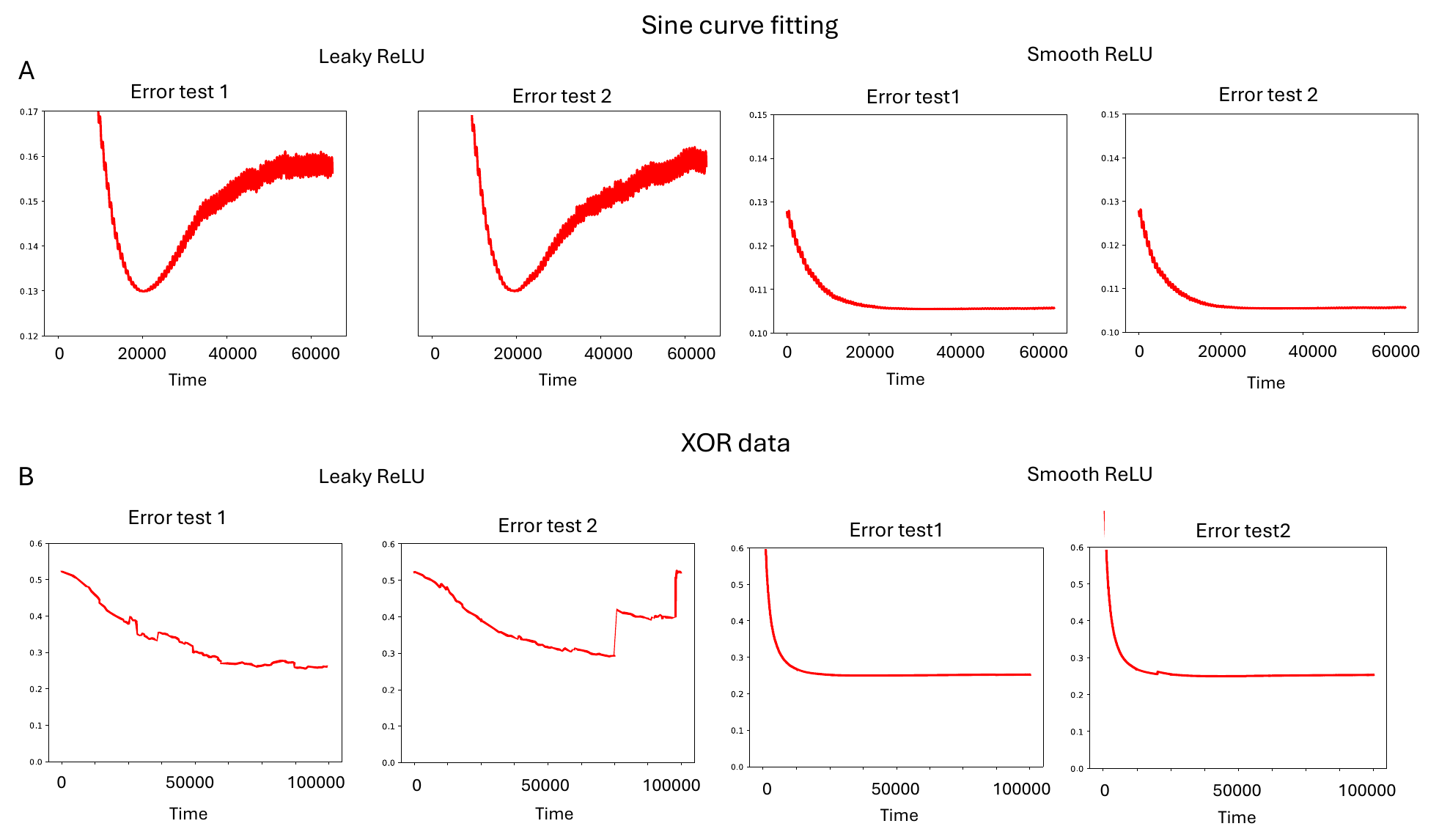}
    \caption{Time evolution of the error functions \eqref{eq:error} obtained by the CRNs implementing NNs with the leaky ReLU and the smooth ReLU. For the sine fitting (A), we used $h=4$ and the noise level $\epsilon=0.1$. For the XOR data set (B), we used $h=3$ and the noise level $\epsilon=0.3$.}
    \label{figs1}
\end{figure}

\begin{figure}[!h]
    \centering
    \includegraphics[width=1\linewidth]{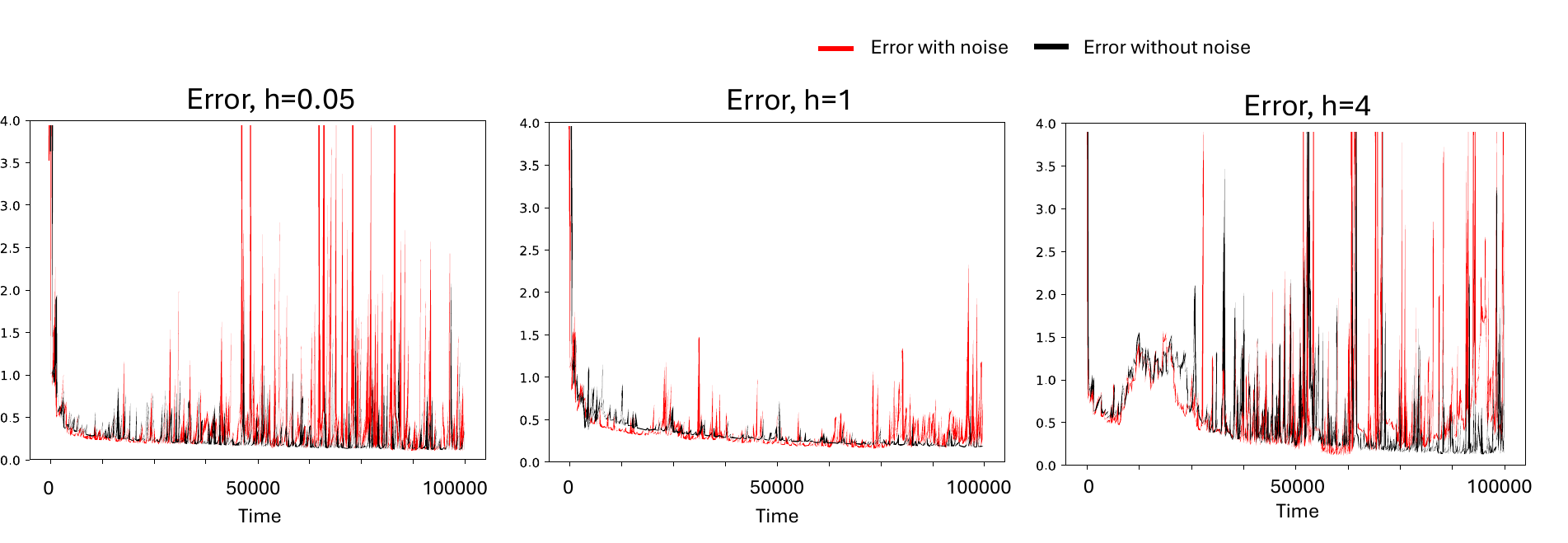}
    \caption{The time evolution of the error function \eqref{eq:error} obtained by the CRN implementing a NN with the smooth ReLU with various values of the smoothing parameter $h$. We used the MNIST data set with the noise level $\epsilon=0.1$.}
    \label{figs2}
\end{figure}

\begin{figure}[!h]
    \centering
    \includegraphics[width=0.8\linewidth]{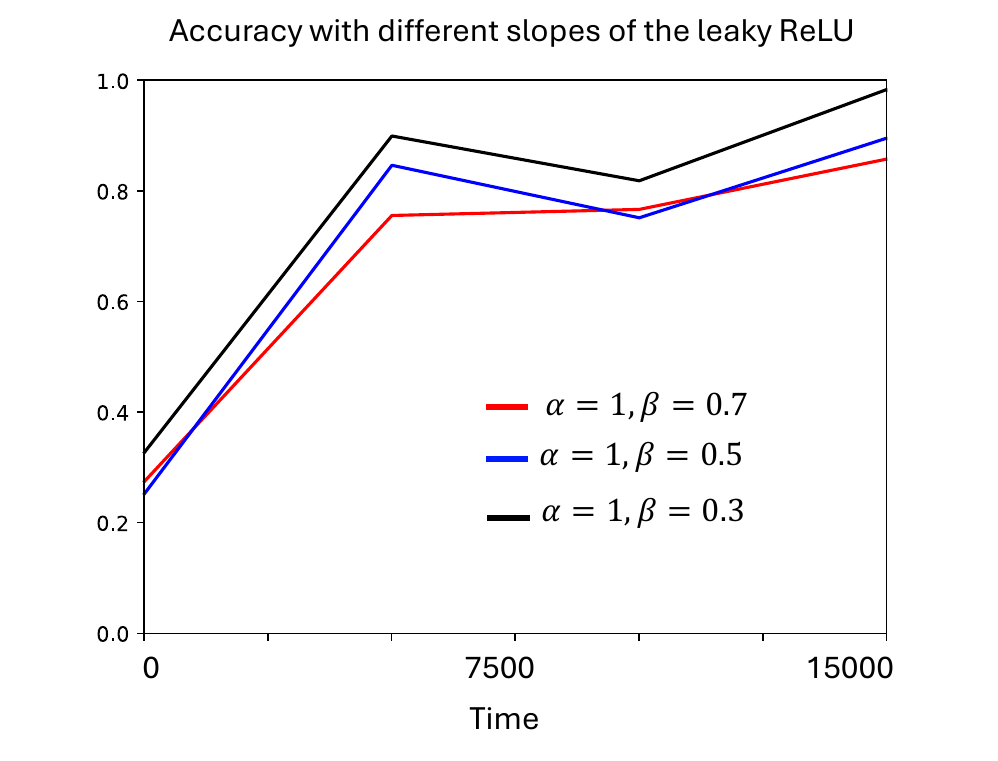}
    \caption{Performance of the CRN implementing a NN using the leaky ReLU with varying the slopes $\alpha$ and $\beta$.  }
    \label{figs3}
\end{figure}

\end{document}